\documentclass[11pt,letterpaper]{article}

\usepackage[margin=1in]{geometry}
\usepackage{theorem,latexsym,graphicx,amssymb}
\usepackage{amsmath,enumerate}
\usepackage{float}
\usepackage{xspace}
\usepackage{paralist}
\usepackage{enumerate}
\usepackage{cases}
\usepackage{caption}
\usepackage{algorithm}
\usepackage[noend]{algpseudocode}
\usepackage{xcolor}
\usepackage{multicol}
\usepackage{graphicx}
\usepackage{subcaption}
\usepackage{varwidth}
\usepackage{wrapfig}

\newenvironment{proof}{{\bf Proof:  }}{\hfill\rule{2mm}{2mm}\vspace*{5pt}}
\newenvironment{proofof}[1]{{\vspace*{5pt} \noindent\bf Proof of #1:  }}{\hfill\rule{2mm}{2mm}\vspace*{5pt}}
\numberwithin{figure}{section}
\numberwithin{equation}{section}
\newtheorem{definition}{Definition}[section]
\newtheorem{remark}{Remark}[section]
\newtheorem{corollary}{Corollary}[section]
\newtheorem{theorem}{Theorem}[section]
\newtheorem{lemma}{Lemma}[section]

\newtheorem{fact}{Fact}[section]

\newcommand{\ranking}{\textsf{Ranking}\xspace}
\newcommand{\mrg}{\textsf{MRG}\xspace}
\newcommand{\rdt}{\textsf{RDO}\xspace}
\newcommand{\irp}{\textsf{IRP}\xspace}
\newcommand{\eqdef}{\stackrel{\textrm{def}}{=}}
\newcommand{\expect}[2]{\underset{#1}{\operatorname{\mathbf E}}\left[#2\right]}
\newcommand{\vect}[1]{\ensuremath{\vec{#1}}}
\newcommand{\vecy}{\vect{y}}

\newcommand{\tbomb}{\textsf{Double-Bomb}\xspace}

\title{Towards a Better Understanding of Randomized Greedy Matching\thanks{This work is supported by National Natural Science Foundation of China (NSFC) 61902233.
The research leading to these results has received funding from the European Research Council under the European Community’s Seventh Framework Programme (FP7/2007-2013) / ERC grant agreement No. 340506. }}

\author{Zhihao Gavin Tang\thanks{ITCS, Shanghai University of Finance and Economics. {\texttt{tang.zhihao@mail.shufe.edu.cn}}}
	\and Xiaowei Wu\thanks{IOTSC, University of Macau. {\texttt{xiaoweiwu@um.edu.mo}}. Part of the work was done when the author was a postdoc at the University of Vienna.}
	\and Yuhao Zhang\thanks{Department of Computer Science, The University of Hong Kong. {\texttt{yhzhang2@cs.hku.hk}}}}
\date{}

\begin{document}
	
\begin{titlepage}
	\thispagestyle{empty}
	\maketitle	
	\begin{abstract}
		There has been a long history for studying randomized greedy matching algorithms since the work by Dyer and Frieze~(RSA 1991). We follow this trend and consider the problem formulated in the oblivious setting, in which the algorithm makes (random) decisions that are essentially oblivious to the input graph.
		
		We revisit the \textsf{Modified Randomized Greedy (\mrg)} algorithm by Aronson et al.~(RSA 1995) that is proved to be $(0.5+\epsilon)$-approximate.
		In particular, we study a weaker version of the algorithm named \textsf{Random Decision Order (\rdt)} that in each step, randomly picks an unmatched vertex and matches it to an arbitrary neighbor if exists.
		We prove the \rdt algorithm is $0.639$-approximate and $0.531$-approximate for bipartite graphs and general graphs respectively. As a corollary, we substantially improve the approximation ratio of \mrg.
		
		Furthermore, we generalize the \rdt algorithm to the edge-weighted case and prove that it achieves a $0.501$ approximation ratio. This result solves the open question by Chan et al.~(SICOMP 2018) about the existence of an algorithm that beats greedy in this setting. As a corollary, it also solves the open questions by Gamlath et al.~(SODA 2019) in the stochastic setting.
	\end{abstract}
\end{titlepage}

\section{Introduction} \label{sec:intro}
Maximum matching is a fundamental problem in combinatorial optimization. Although the problem admits an efficient polynomial time algorithm, the greedy heuristic is widely used and observed to have good performance~\cite{anor/Tinhofer84,jea/Magun98}.
Since the initial work by Dyer and Frieze~\cite{rsa/DyerF91} in the early-nineties, computer scientists have been interested in the worst-case performance of (randomized) greedy algorithms.
In addition to this pure theoretical interest, greedy algorithms have also attracted attention due to kidney exchange applications~\cite{jet/RothSU05}. 
In such scenarios, information about the graph is incomplete and greedy algorithms are the only algorithms one can implement.

To this end, the \emph{oblivious matching} model is formulated~\cite{focs/GoelT12, sicomp/ChanCWZ18}.
Consider a graph $G=(V,E)$ in which the vertices $V$ are revealed while the edges $E$ are unknown.
The algorithm picks a permutation of unordered pairs of vertices $\binom{V}{2}$.
Each pair is probed one-by-one according to the permutation to form a matching greedily.
In particular, when the pair $(u,v)$ is probed, if the edge exists and both $u,v$ are unmatched, then we match the two vertices; otherwise, we continue to the next pair.
That is, the algorithm works under the query-commit model.
We compare the performance of an algorithm to the size of a maximum matching.

Note that any permutation induces a maximal matching and, hence, any algorithm is $0.5$-approximate. On the other hand, no deterministic algorithms can do better than this ratio. The interesting question is then to design a randomized algorithm with approximation ratio greater than $0.5$.

\subsection{Prior Works}
The first natural attempt for randomization is to permute all pairs of vertices uniformly at random. Unfortunately,  Dyer and Frieze~\cite{rsa/DyerF91} proved that it fails to beat the $0.5$ approximation ratio.
The first non-trivial theoretical guarantee for the problem is provided by Aronson et al.~\cite{rsa/AronsonDFS1995}. In the paper, they proposed the \textsf{Modified Randomized Greedy (MRG)} algorithm and proved a $(0.5+\epsilon)$\footnote{$\epsilon=1/400000$.} lower bound on its approximation ratio. The analysis of \mrg was later improved by Poloczek and Szegedy~\cite{focs/PoloczekS12} to $0.504$.\footnote{It was pointed out by Chan et al.~\cite{sicomp/ChanCWZ18} that their paper contains some gaps in the proof.} Different randomized greedy algorithms are also studied, including \ranking~\cite{stoc/KarpVV90, stoc/KarandeMT11, stoc/MahdianY11, sicomp/ChanCWZ18}, \textsf{FRanking}~\cite{stoc/HKTWZZ18, soda/HPTTWZ19}, etc.

Interestingly, all existing algorithms fall into the family of vertex-iterative algorithms~\cite{focs/GoelT12}:
\begin{enumerate}
	\item Iterates through the vertices according to a \emph{decision order}.
	\item In the iteration of $u$, probe $u$ with other vertices according to the \emph{preference order} of $u$.
\end{enumerate}
There are $|V|+1$ different orders to be specified, including a decision order in the first step and $|V|$ individual preference orders in the second step.
The names are chosen due to the following equivalent statement of the framework: 1) let the vertices make decisions sequentially according to the decision order; 2) if a vertex is not matched at its decision time, it would choose its favorite unmatched neighbor, according to its individual preference.

Existing algorithms differ in the way the orders are generated.
In particular, \mrg samples the decision order and the preference orders independently and uniformly at random. \ranking~\cite{stoc/KarpVV90} samples only one random permutation uniformly and uses it as the decision order and the common preference order. It was observed~\cite{focs/GoelT12,focs/PoloczekS12} that the approximation ratio of \ranking for the oblivious matching problem on bipartite graphs is the same as for online bipartite matching with random arrival order.
Hence the result of Mahdian and Yan~\cite{stoc/MahdianY11} translates to a $0.696$-approximation in the oblivious setting for bipartite graphs. For general graphs, \ranking is proved to be $0.526$-approximate~\cite{sicomp/ChanCWZ18, talg/ChanCW18}\footnote{Goel and Tripathi (FOCS 2012) claimed that \ranking is 0.56-approximate, but later withdrew the paper when they discovered	a bug in their proof.}.
Recently, the \textsf{FRanking} algorithm has been studied in the fully online matching model~\cite{stoc/HKTWZZ18, soda/HPTTWZ19}\footnote{In their paper, the algorithm is named \ranking. We use \textsf{FRanking} to distinguish it from the original \ranking algorithm, since they have different behavior when adapted to the oblivious matching setting.}. 
Their results can be transformed to an oblivious matching algorithm that uses an arbitrary decision order\footnote{The feature of arbitrary decision time comes from the online nature of their setting.} and a common random preference order.
The competitive ratios $0.567$ and $0.521$ achieved for bipartite graphs~\cite{soda/HPTTWZ19} and general graphs~\cite{stoc/HKTWZZ18} are then inherited in our setting.

\begin{table}[h]
	\centering
	\begin{tabular}{|c|c|c|c|c|}
		\hline
		& Decision         & Preferences         & Bipartite & General          \\ \hline
		\ranking    & \multicolumn{2}{c|}{Common Random} & $0.696$~\cite{stoc/MahdianY11}           & $0.526$~\cite{sicomp/ChanCWZ18, talg/ChanCW18}   \\ \hline
		\textsf{FRanking}    & Arbitrary       & Common Random & $0.567$~\cite{soda/HPTTWZ19}           & $0.521$~\cite{stoc/HKTWZZ18}         \\ \hline
		\mrg         &
		\multicolumn{2}{c|}{Independent Random} & \multicolumn{2}{c|}{$0.5+\epsilon$~\cite{rsa/AronsonDFS1995,focs/PoloczekS12} $\to {\bf 0.531}$ (Section~\ref{sec:unweighted_gen})} \\ \hline
		{\textsf{IRP}} & Arbitrary	& Independent Random & \multicolumn{2}{c|}{$0.5$~\cite{rsa/DyerF91}} \\ \hline
		{\bf \rdt} & Random       & Arbitrary & ${\bf 0.639}$~(Section~\ref{sec:unweighted_bipartite})           & ${\bf 0.531}$~(Section~\ref{sec:unweighted_gen})          \\ \hline
	\end{tabular}
	\caption{A summary of prior works and our results. Our results are marked in bold. ``Common random'' means a single random permutation is used for all orders. ``Independent random'' means the orders are generated independently.}
	\label{table:prior_work}
\end{table}

\subsection{Our Contributions}
Despite the successful progresses of these results, we lack a systematic understanding on the roles of the decision order and preference orders. 
In this paper, we revisit the \mrg algorithm and study separately the two kinds of randomness on the decision order and the preference orders.

Consider a variant of the \mrg algorithm in which the decision order is fixed arbitrarily while the preference orders are drawn independently and uniformly for all vertices. We refer to this algorithm as \textsf{Independent Random Preferences (IRP)}. Due to the instance from~\cite{rsa/DyerF91}, \irp is not better than $0.5$-approximate.\footnote{This is observed in the online matching literature. For completeness, we include a formal discussion in Appendix~\ref{app:hardness}.}
In other words, the random decision order is necessary and crucial for \mrg to work. On the other hand, it is not clear from the previous analysis~\cite{rsa/AronsonDFS1995, focs/PoloczekS12} whether the random decision order alone is sufficient.

In this paper, we answer this question affirmatively and conclude that the randomness of decision order plays a more important role than the preference orders for the \mrg algorithm. Consider the following variant of the \mrg algorithm.

\paragraph{Random Decision Order (RDO)} We sample the decision order uniformly at random and fix the individual preference orders arbitrarily.

\begin{theorem}
	\label{th:unweighted_bipartite}
	\rdt is $0.639$-approximate for the oblivious matching problem on bipartite graphs.
\end{theorem}

\begin{theorem}
	\label{th:unweighted_general}
	\rdt is $0.531$-approximate for the oblivious matching problem on general graphs.
\end{theorem}

As an immediate corollary, the approximation ratios apply to \mrg, that substantially improve the $(0.5+\epsilon)$ ratios by Aronson et al.~\cite{rsa/AronsonDFS1995} and Poloczek and Szegedy~\cite{focs/PoloczekS12} for general graphs.
Besides, our result beats the state-of-the-art $0.526$-approximation by \ranking~\cite{talg/ChanCW18} for the oblivious matching problem.

\begin{corollary} \label{corollary:mrg}
	\mrg is $0.531$-approximate for the oblivious matching problem.
\end{corollary}

We also show some hardness results for the \rdt algorithm.
First, we give a $0.646$ upper bound on the approximation ratio of \rdt on bipartite graphs by experiments, which is very close to our $0.639$ lower bound.
Second, we prove that \rdt is at most $0.625$-approximate on general graphs. Together with Theorem~\ref{th:unweighted_bipartite}, it gives a separation on the approximation ratio of \rdt on bipartite and general graphs. To the best of our knowledge, no such separation has been shown to exist for other algorithms including \mrg, \ranking and \textsf{FRanking}.

\paragraph{Extension to Edge-weighted Case.} Interestingly, the better understanding of \mrg towards \rdt lends us to a natural generalization of the algorithm to the edge-weighted oblivious matching problem. In the edge-weighted case, the graph is associated with a weight function $w$ defined on all pairs of vertices that is known by the algorithm. If an edge $(u,v)$ exists, its weight is given by $w_{uv}$.
It is straightforward to check that the greedy algorithm that probes edges in descending order of the weights is $0.5$-approximate.
It has been proposed as an open question in \cite{sicomp/ChanCWZ18} to study the existence of a better than $0.5$-approximate algorithm.
In this paper, we answer it positively by studying the following generalization of \rdt.

\paragraph{Perturbed Greedy.} 
Each vertex $u\in V$ draws a rank $y_u\in [0,1]$ independently and uniformly.
Then probe all pairs of vertices $(u,v)$ in descending order of their perturbed weights\footnote{Break ties arbitrarily and consistently.}, which is defined as $(1-g(\min\{y_u,y_v\}))\cdot w_{uv}$, where $g$ is a non-decreasing function to be fixed later.

\begin{theorem}\label{th:weighted_general}
	There exists a function $g$ so that \textsf{Perturbed Greedy} is $0.501$-approximate for the edge-weighted oblivious matching problem.
\end{theorem}

It is easy to show that the \textsf{Perturbed Greedy} algorithm degenerates to \rdt for unweighted graphs, for any increasing function $g$. Indeed, consider $y_u$ as the \emph{decision time} of vertex $u$. We list all vertices in ascending order of their decision times. This is equivalent to sampling a random decision order. Upon the decision of $u$, all remaining edges incident to $u$ would have the same smallest perturbed weight $1-g(y_u)$. By breaking ties arbitrarily and consistently, it is equivalent to study arbitrary individual preference orders, i.e. the \rdt algorithm.

\paragraph{Extension to Stochastic Probing.}
Another well-studied maximum matching model is the \emph{stochastic probing} problem~\cite{icalp/ChenIKMR09,icalp/CostelloTT12,soda/GamlathKS19}.
In addition to the edge-weighted oblivious setting, each edge is associated with an existence probability that is known by the algorithm.
It is assumed that the underlying graph is generated by sampling each edge independently with the existence probability.

Very recently, Gamlath et al.~\cite{soda/GamlathKS19} proposed a $(1-\frac{1}{e})$-approximate algorithm for the problem on bipartite graphs, which is the first result bypassing the $0.5$ barrier.
They proposed\footnote{The open questions are raised in their SODA 2019 conference presentation.} two open questions: (1) does algorithm with approximation ratio strictly above $0.5$ exist on general graphs? (2) what approximation ratio can be obtained if the existence probabilities of edges are correlated?

Observe that any algorithm in the oblivious setting can be applied in the stochastic setting by ignoring the extra stochastic information. Moreover, the approximation ratio is preserved since it holds for every realization of the underlying graph, even when the edges are sampled from correlated distributions.
Thus, as a corollary of Theorem~\ref{th:weighted_general}, we answer both questions positively.
\begin{corollary}\label{th:weighted_general_stochastic}
	There exists an algorithm for the edge-weighted stochastic probing problem on general graphs that is $0.501$-approximate, even when the existence probabilities of edges are correlated.
\end{corollary}

\subsection{Our Techniques}

We will prove our results starting with \rdt on bipartite graphs, then \rdt on general graphs and finally \textsf{Perturbed Greedy} on edge-weighted general graphs, in progressive order of difficulty.
Our analysis builds on the randomized primal dual framework introduced by Devanur et al.~\cite{soda/DevanurJK13} and recently developed in a sequence of results~\cite{stoc/HKTWZZ18, icalp/HTWZ18, soda/HPTTWZ19}.
Roughly speaking, we split the gain of each matched edge to its two endpoints.
By proving that the expected combined gain of any pair of neighbors $u,v$ is at least $r\cdot w_{uv}$, we have that the approximation ratio is at least $r$.
However, our approach differs from previous works in a way that we only look at the pairs that appear in some fixed maximum matching.
We refer to such pairs as perfect partners. In order to prove the algorithm is $r$-approximate, we observe that it suffices to show the expected combined gain of $(u,v)$ is at least $r\cdot w_{uv}$ for all perfect pairs $(u,v)$.
As far as we know, our result is the first successful application of the randomized primal dual technique to the oblivious matching problem, or to any edge-weighted maximum matching problems.

\paragraph{Unweighted Bipartite Graphs.}
As discussed in the previous subsection, the random decision order can be generated by sampling decision times independently and uniformly from $[0,1]$ for all vertices and then list the vertices in ascending order of their decision times. We shall use this interpretation for the purpose of analysis. We use $y_u$ to denote the decision time of vertex $u$, and it plays a similar role as the rank of $u$ in the \ranking algorithm in our analysis.

Consider a pair of neighbors $(u,v)$, where $u$ has an earlier decision time than $v$.
Existing analysis on \ranking relies on an important structural property that whenever $v$ is unmatched, $u$ must be matched to some vertex with rank smaller than $y_v$.
However, we have nothing to say about $u$'s choice on its decision time, since we assume it makes the decision based on its individual preference. %
To this end, whenever an edge is matched, we split the gain
between the endpoints based on the current decision time instead of the rank/decision time of the passively chosen vertex.
We specify such a gain sharing rule so that by fixing the decision times of all vertices other than $u,v$ arbitrarily and taking the randomness over $y_u,y_v$, the expected combined gain of $u$ and $v$ is at least $0.639$.

\paragraph{Unweighted General Graphs.}
For bipartite graphs, our analysis relies on a crucial structural property that if $u$ is matched when its neighbor $v$ is removed from the graph, then $u$ remains matched when $v$ is added back, no matter what the decision time $v$ has.
Going from bipartite graphs to general graphs takes away this nice property.
A similar issue arises in the analysis of \ranking, \textsf{FRanking}.
The recent work by Huang et al.~\cite{stoc/HKTWZZ18} introduce a notion of \emph{victim} to tackle the problem.
Roughly speaking, they call an unmatched vertex $v$ the victim of its neighbor $z$ if $v$ becomes matched when $z$ is removed.
Then they define a compensation rule in which each active vertex sends an amount of gain, named compensation, to its victim (if any).

In this work, we propose a new definition of victim together with a compensation rule that is arguably more intuitive and has a clearer structure compared to that of~\cite{stoc/HKTWZZ18}.

Suppose vertex $z$ actively matches $u$ in our algorithm.
If $u$ is matched with its \emph{perfect partner} $v$ when $z$ is removed, we call $v$ the \emph{victim} of $z$.
After sharing the gain as we have described for bipartite graphs, let each active vertex send a portion of its gain to its victim if the victim is unmatched. This compensation rule is designed to retrieve some extra gain for $u,v$ when the aforementioned property for bipartite graphs fails to hold.

\paragraph{Weighted General Graphs.}
For unweighted graphs, the contribution of each matched vertex is fixed. Therefore, it suffices to summarize the status of a vertex as matched or unmatched.
However, being matched is no longer a meaningful signature on edge-weighted graphs.
E.g., being matched in an $\epsilon$-weighted edge and being matched in a large weight edge should be treated differently.
This observation prevents the victim notion of Huang et al.~\cite{stoc/HKTWZZ18} from generalizing to edge-weighted case.

More specifically, under their notion, the victim $v$ of $z$ will become matched when $z$ is removed.
However, there is no guarantee on the weight of edge $v$ matches.
In contrast, we define $v$ the victim of $z$ only if $v$ matches its perfect partner $u$ when $z$ is removed.

Consequently, our definition of victim generalizes naturally to the weighted case, and the analysis for unweighted graphs extends with a minimum change of the compensation rule.
In particular, we fix a function $h$ that represents the amount of compensation one would like to send. The compensation rule works in the following way. Suppose $v$ is the victim of $z$. We know that $z$ is matched with $v$'s perfect partner $u$. Our compensation rule ensures that after the compensation, $v$ has gain at least $h(y_z)\cdot w_{zu}$.
Note that the amount of compensation that $z$ sends depends on the current gain of $v$.
This generalization further justifies the advantage of our new notion of victim and we believe the new notion will find more applications in other matching problems on general graphs.

\subsection{Other Related Works}
For the hardness results of the oblivious matching problem, Goel and Tripathi~\cite{focs/GoelT12} show a $0.7916$ upper bound on the approximation ratio of any algorithm for unweighted graphs.
When restricted to the family of vertex iterative algorithms, they give a stronger bound of $0.75$.
For the \mrg algorithm, Dyer and Frieze~\cite{rsa/DyerF91}'s bomb graphs give a $2/3$ upper bound on its approximation ratio.
For the \ranking algorithm, the hard instance provided by Mahdian and Yan~\cite{stoc/MahdianY11} implies a $0.727$ upper bound, which was later improved to $0.724$ by Chan et al.~\cite{sicomp/ChanCWZ18}. 

Independent from our work, Huang~\cite{corr/Huang19} studied the edge-weighted online bipartite matching problem using the randomized primal-dual technique. They provide a simplification of the algorithm by Zadimoghaddam~\cite{corr/Zadimoghaddam17}, and show that it obtains an improved competitive ratio of $0.514$.

\section{Preliminaries}\label{sec:prelim}

Recall the general algorithm for the edge-weighted oblivious matching problem as follows.

\begin{algorithm}
	\caption{\textsf{Perturbed Greedy}}
	\label{alg:rog}
	\begin{algorithmic}
		\State Fix a non-decreasing function $g:[0,1]\to[0,1]$.
		\State Each vertex $u$ independently draws a rank $y_u \in [0,1]$ uniformly at random.
		\State Probe pairs in descending order of their perturbed weights $(1-g(\min \{y_u, y_v\})) w_{uv}$. 
	\end{algorithmic}
\end{algorithm}

We use $\vec{y}$ to denote the ranks of vertices, and $M(\vec{y})$ to denote the matching produced by our algorithm with ranks $\vec{y}$.
We use $M_{u}(\vec{y})$ to denote the matching produced by our algorithm on $G-\{u\}$, i.e. when vertex $u$ is removed from the graph.
Fix any maximum weight matching $M^*$, the approximation ratio is the ratio between the expected total weight of edges in $M(\vecy)$, and the total weight of edges in $M^*$.

The gain sharing framework we use in this paper is formalized as follows.

\begin{lemma}\label{lemma:dual}
	If there exist non-negative random variables $\{ \alpha_u \}_{u\in V}$ depending on $\vecy$ such that
	\begin{compactitem}
		\item for every $\vecy$, $\sum_{u\in V} \alpha_u = \sum_{(u,v)\in M(\vecy)}w_{uv}$;
		\item for every $(u,v)\in M^*$, $\expect{\vecy}{\alpha_u + \alpha_v} \geq r\cdot w_{uv}$,
	\end{compactitem}
	then our algorithm is $r$-approximate.
\end{lemma}
\begin{proof}
	The approximation ratio of \textsf{Perturbed Greedy} is given by
	
	$\frac{\expect{\vecy}{\sum_{(u,v)\in M(\vecy)} w_{uv}}}{\sum_{(u,v)\in M^*} w_{uv}}
	= \frac{\expect{\vecy}{\sum_{u \in V} \alpha_u}}{\sum_{(u,v)\in M^*} w_{uv}}
	\geq \frac{\expect{\vecy}{\sum_{(u,v)\in M^*} (\alpha_u + \alpha_v)}}{\sum_{(u,v)\in M^*} w_{uv}}
	\geq \frac{\sum_{(u,v)\in M^*}r\cdot w_{uv}}{\sum_{(u,v)\in M^*} w_{uv}} = r$.
\end{proof}

Next, we define the notion of active and passive vertices for edge-weighted general graphs and provide the monotonicity on the ranks (the proof can be found in Appendix~\ref{appendix:prelim}).

\begin{definition}[Active, Passive]
	For every edge $(u,v)$ added to the matching by \textsf{Perturbed Greedy} with $y_u < y_v$, we say that $u$ is active and $v$ is passive.
\end{definition}

\begin{lemma}[Monotonicity]\label{lemma:monotonicity}
	Consider any matching $M(\vec{y})$ and any vertex $v$.
	If $v$ is passive or unmatched, then there exists some threshold $y \leq y_v$ such that if we reset $y_v$ to be any value in $(y,1)$, the matching remains unchanged; if we reset $y_v$ to be any value in $(0, y)$, then $v$ becomes active.
	Moreover, the weight of the edge $v$ actively matches is non-increasing w.r.t. $y_v\in (0,y)$.
\end{lemma}

\subsection{Unweighted Graphs.}
In this subsection, we focus on unweighted graphs and provide the standard \emph{alternating path} property.
An analog for edge-weighted graphs will be provided in Section~\ref{sec:weighted_general}.

As we have discussed in Section~\ref{sec:intro}, we can equivalently interpret the algorithm as follows.

We call the rank $y_u$ of vertex $u$ as the \emph{decision time} of $u$, and let vertices act in ascending order of their decision times.
If a vertex is not matched yet at its decision time, then it will match the unmatched neighbor according to its own preference order.
We also call $\min\{y_u,y_v\}$ the decision time of edge $(u,v)\in E$, which is the only possible time the edge is included in the matching.

\begin{algorithm}[H]
	\caption{\textsf{Random Decision Order}}
	\label{alg:rog_time}
	\begin{algorithmic}
		\State Each vertex $u$ independently draws a rank $y_u \in [0,1]$ uniformly at random.
		\State At decision time $y_u$ of $u$, if $u$ is unmatched, $u$ chooses its favourite unmatched neighbor. 
	\end{algorithmic}
\end{algorithm}

Note that for each vertex, the preference order of its neighbors is arbitrary but fixed for each vertex. In other words, it does not depend on the ranks $\vec{y}$ of vertices.

It is easy to see that for the unweighted case, if a vertex $v$ is passively matched by $u$, then the threshold of $v$ in Lemma~\ref{lemma:monotonicity} coincides with the rank $y_u$ of $u$. Thus we have the following.

\begin{corollary}[Unweighted Monotonicity]\label{corollary:monotonicity_unweighted}
	Consider any matching $M(\vec{y})$ and any vertex $v$.
	If $v$ is passively matched by $u$, then when $v$ has decision time in $(y_u,1)$, the matching remains unchanged; when $v$ has decision time in $(0, y_u)$, then $v$ becomes active.
	If $v$ is active, then the set of unmatched neighbors $v$ sees at its decision time grows when $y_v$ decreases.
\end{corollary}

The following important property characterizes the effect of the removal of a single vertex.
For continuity of presentation we defer the proof of the lemma to Appendix~\ref{appendix:prelim}.

\begin{lemma}[Alternating Path]\label{lemma:alternating}
	If $u$ is matched in $M(\vecy)$, then the symmetric difference between $M(\vecy)$ and $M_{u}(\vecy)$ is an alternating path $(u_0=u,u_1,u_2,\ldots)$ such that for all even $i$, $(u_i,u_{i+1}) \in M(\vecy)$ and for all odd $i$, $(u_i, u_{i+1}) \in M_u(\vecy)$.
	Moreover, the decision times of edges along the path are non-decreasing.
	Consequently, vertices $u_1,u_3,\ldots$ are matched no later in $M(\vecy)$ than in $M_{u}(\vecy)$.
\end{lemma}

Given the above lemma, we show the following useful property, which roughly says that any vertex $v$ can not affect the matching status of another vertex $u$ if $u$ is matched before $y_v$.

\begin{corollary}\label{corollary:limit_of_affect}
	Suppose at time $y$, vertex $u$ is matched while vertex $v$ is still unmatched.
	Then resetting $y_v$ to be any value in $(y,1)$ does not change the matching status of $u$.
\end{corollary}
\begin{proof}
	Since $v$ is unmatched when $u$ gets matched, if $v$ is removed, then the matching status of $u$ is not affected.
	Now suppose that we insert $v$ with $y_v \in (y,1)$, then if $v$ is matched, then it must be matched at time later than $y$ (the time when $u$ gets matched).
	In other words, the insertion of $v$ triggers a (possibly empty) alternating path in which all edges (and therefore vertices) have decision time at most $y$.
	Hence, $u$ is not included in the alternating path and its matching status is not affected.
\end{proof}

For bipartite graphs, Lemma~\ref{lemma:alternating} also implies the following nice property.

\begin{corollary}\label{corollary:insertion_bipartite}
	For bipartite graphs, if $u$ is matched in $M_{v}(\vecy)$ and $v$ is a neighbor of $u$, then $u$ is also matched in $M(\vecy)$. Moreover, the time $u$ is matched in $M(\vecy)$ is no later than in $M_{v}(\vecy)$.
\end{corollary}
\begin{proof}
	By Lemma~\ref{lemma:alternating}, inserting $v$ (at any rank) creates a (possibly empty) alternating path $(v,v_1,v_2,\ldots)$.
	As the graph is bipartite, if $u$ appears in the path, then it must be one of $\{v_1,v_3,\ldots\}$, which is matched no later than in $M_{v}(\vecy)$.
	Otherwise its matching status is not affected.
\end{proof}

\section{Unweighted Bipartite Graphs}\label{sec:unweighted_bipartite}

Recall that for unweighted case, vertices act in ascending order of their decision times, and each unmatched vertex chooses its favorite unmatched neighbor.
We first give the gain sharing rule for every matched pair $(u,v)\in M(\vecy)$.

\paragraph{Gain Sharing.} Whenever $u$ actively chooses $v$ at time $y_u$, let $\alpha_u=g(y_u)$ and $\alpha_v=1-g(y_u)$, where $g$ is a non-decreasing function to be fixed later.

\paragraph{General Framework.}
Recall that to prove an approximation ratio of $r$, it suffices to show that $\expect{\vecy}{\alpha_u + \alpha_v} \geq r$ for every $(u,v)\in M^*$.
Fix such a pair $(u,v)$, and fix the decision times of all vertices other than $u,v$ arbitrarily.
For ease of notation, we use $M(y_u,y_v)$ to denote the matching produced by our algorithm when $u,v$'s decision times are $y_u$ and $y_v$, respectively.
In the following, we will give a lower bound $f(y_u,y_v)$ of $\alpha_u + \alpha_v$ for each $M(y_u, y_v)$, and show that there exists an appropriate function $g$ such that $\int_0^1 \int_0^1 f(y_u, y_v) dy_u dy_v \geq 0.639$, finishing the proof of Theorem~\ref{th:unweighted_bipartite}.

\medskip

Consider matching $M(1,1)$, i.e., when $u,v$ have the latest decision times compared with other vertices. 
Depending on whether $u,v$ are matched together, we divide our analysis into two parts.

\subsection{Symmetric Case: $(u,v)\in M(1,1)$}

In this case, $u$ and $v$ are not chosen by any other vertex. At time $1$, $u$ and $v$ are matched together.

\begin{lemma}
	If $(u,v)\in M(1,1)$, then when $y_u < y_v$, $u$ is active in $M(y_u,y_v)$ and $v$ is unmatched before time $y_u$; when $y_v < y_u$, $v$ is active in $M(y_u,y_v)$, and $u$ is unmatched before time $y_v$.
\end{lemma}
\begin{proof}
	Suppose $y_u < y_v$.
	Consider the first time $t$ when one of $u, v$ (say, $x\in \{u,v\}$) is matched.
	Obviously, $t \le y_u$.
	If $t < y_u$, then $x$ must be chosen by a vertex $z$ at its decision time $y_z=t$.
	By Lemma~\ref{lemma:monotonicity}, we have $(x,z)\in M(1,1)$, which violates the assumption of the lemma.
	Thus $t = y_u$, i.e, $u$ is active, and $v$ is unmatched before time $y_u$.
	The case when $y_v < y_u$ is similar.
\end{proof}

We decrease $y_u$ gradually from $1$ to $0$ and study $M(y_u,1)$.
By Corollary~\ref{corollary:monotonicity_unweighted}, the set of unmatched neighbors of $u$ at time $y_u$ grows when $y_u$ decreases.
Hence there exists a transition time $\theta$ such that $v$ is no longer $u$'s favorite vertex when $y_u < \theta$.
In other words, when $y_u > \theta$, $u$ actively matches $v$ in $M(y_u,1)$; when $y_u < \theta$, $u$ actively matches a vertex other than $v$ in $M(y_u,1)$. 
Moreover, by Corollary~\ref{corollary:limit_of_affect}, the matching status of $u$ in $M(y_u,y_v)$ is the same as in $M(y_u,1)$, as long as $y_v > y_u$.
Thus we have (refer to Figure~\ref{fig:bip1})
\begin{compactitem}
	\item $\alpha_u + \alpha_v = 1$ when $y_u > \theta$ and $y_v > y_u$;
	\item $\alpha_u = g(y_u)$ when $y_u < \theta$ and $y_v > y_u$.
\end{compactitem}

Similarly, we decrease $y_v$ gradually from $1$ to $0$ and study $M(1,y_v)$.
Let $\tau$ be the transition time such that $u$ is $v$'s favourite neighbor if and only if $y_v > \tau$.
Then we have (refer to Figure~\ref{fig:bip1})
\begin{compactitem}
	\item $\alpha_u + \alpha_v = 1$ when $y_v > \tau$ and $y_u > y_v$;
	\item $\alpha_v = g(y_v)$ when $y_v < \tau$ and $y_u > y_v$.
\end{compactitem}
\vspace*{-10pt}
\begin{figure}[H]
	\centering
	\begin{subfigure}{.4\textwidth}
		\centering
		\includegraphics[width=0.75\linewidth]{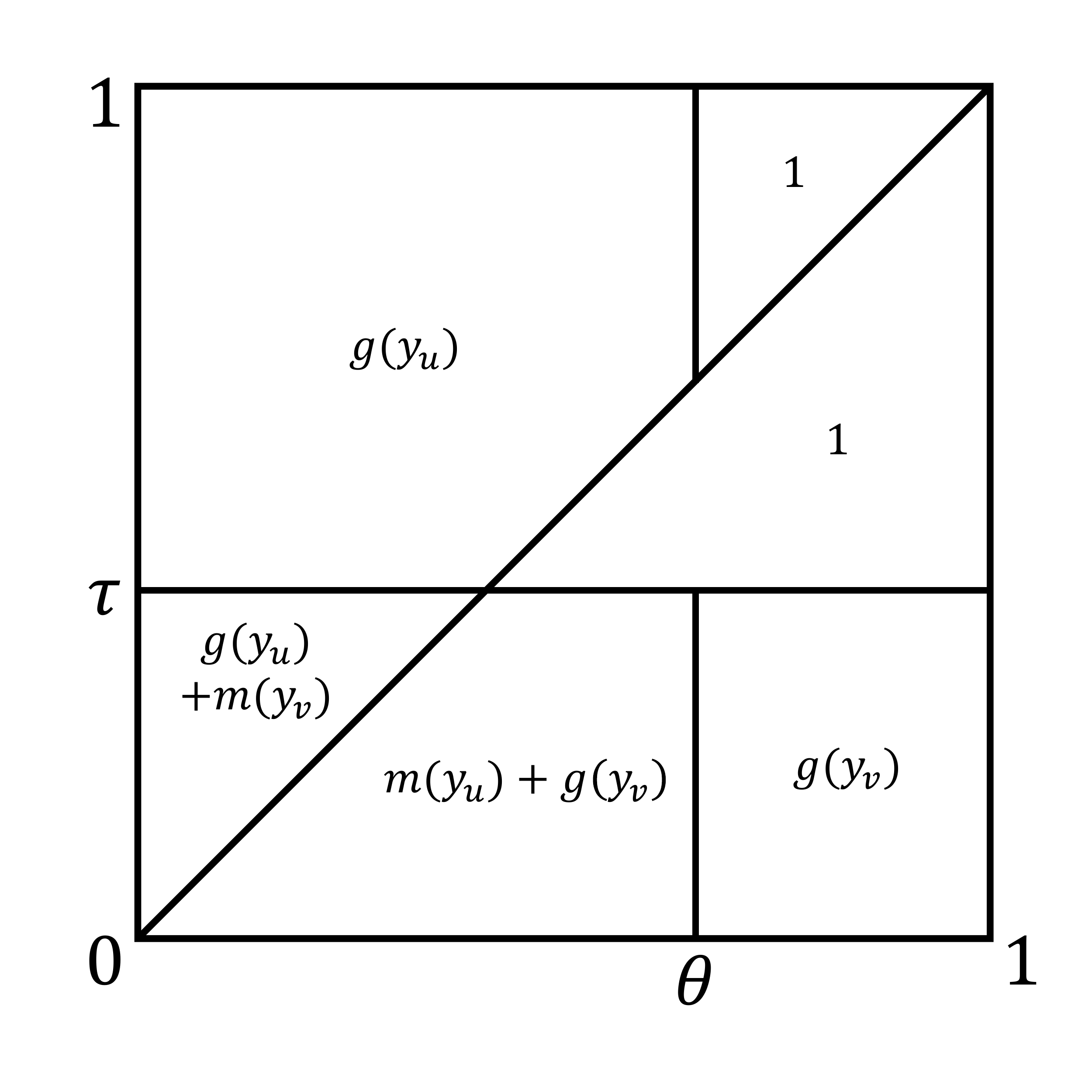}
		\vspace*{-5pt}
		\caption{Symmetric case: $u,v$ match each other}
		\label{fig:bip1}
	\end{subfigure}
	\begin{subfigure}{.4\textwidth}
		\centering
		\includegraphics[width=0.75\linewidth]{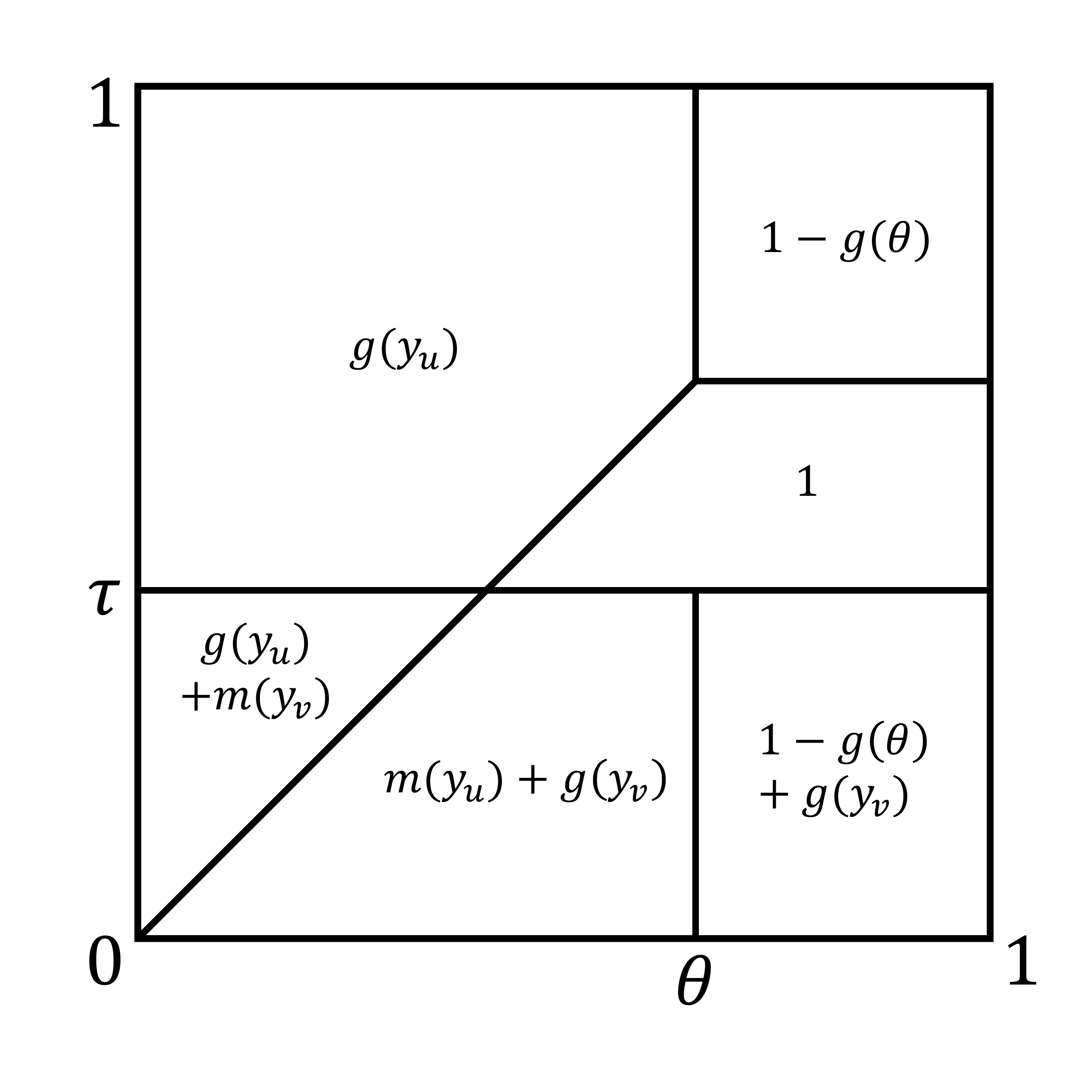}
		\vspace*{-5pt}
		\caption{Asymmetric case: $u$ is chosen at $\theta$}
		\label{fig:bip2}
	\end{subfigure}
	\caption{Unweighted bipartite graphs: the horizontal and vertical axes correspond to $y_u,y_v$ respectively. For each region, the formula written serves as a lower bound of $\alpha_u+\alpha_v$.}
	\label{fig:test}
\end{figure}

We refer to these gains as the basic gain of our analysis as they come immediately after we properly define $\theta,\tau$. Next, we study the matching status of the vertex with later decision time and achieve some extra gains, where we crucially use the bipartiteness of the graph.

\begin{lemma}[Extra Gain]
	For all $y_u < \theta$ and $y_v < \tau$, both $u$ and $v$ are matched in $M(y_u, y_v)$.
\end{lemma}
\begin{proof}
	By definition, when $y_u < \theta$, $u$ actively matches a vertex other than $v$ in $M(y_u, 1)$.
	Thus removing $v$ does not change the matching status of $u$. In other words, $u$ is matched in $M_v(y_u,y_v)$.
	By Corollary~\ref{corollary:insertion_bipartite}, $u$ remains matched in $M(y_u,y_v)$.
	Similarly, we have $v$ is matched in $M(y_u,y_v)$ for all $y_v <\tau$, which finishes the proof.
\end{proof}

Let $m(y) \eqdef \min\{ g(y), 1-g(y) \}$. It is easy to see that whenever $u$ is matched (actively or passively), $\alpha_u \geq m(y_u)$.
In summary, we have the following lower bound (refer to Figure~\ref{fig:bip1}).
\begin{equation} \label{eq:uw_bip_uv}
	\begin{split}
		\expect{y_u,y_v}{\alpha_u + \alpha_v} \geq & \int_0^\theta (1-y_u)g(y_u)dy_u + \frac{1}{2}(1-\theta)^2 + \int_0^\tau (1-y_v)g(y_v) dy_v + \frac{1}{2}(1-\tau)^2 \\
		&  +\int_0^\theta \min\{y_u,\tau\} m(y_u) dy_u + \int_0^\tau \min\{y_v,\theta\} m(y_v) dy_v .
	\end{split}
\end{equation}

\subsection{Asymmetric Case: $(u,v)\notin M(1,1)$}

In this case, at least one of $u,v$ is matched before time $1$.
Without loss of generality, suppose $u$ is matched at time $\theta < 1$, and strictly earlier than $v$.
Observe that $u$ must be passive since $y_u=1$. Let $z$ be the vertex that actively matches $u$.
Then we have $y_z = \theta$.
Intuitively, $u$ is the ``luckier'' vertex compared with $v$ since it is favored by a vertex with early decision time.
Indeed, $u$ would remain matched even when $v$ is removed from the graph.

First, observe that when both $u$ and $v$ have decision times larger than $\theta$, $u$ is always matched by $z$, and thus $\alpha_u = 1-g(\theta)$.
When $y_u < \theta$, $u$ must be active in $M(y_u,1)$, since at time $y_u$, $u$ is unmatched and has unmatched neighbors $z$ and $v$ (with later decision times).
Moreover, by Corollary~\ref{corollary:limit_of_affect}, $u$ is active in $M(y_u,y_v)$ as long as $y_u< \theta$ and $y_v > y_u$. Thus $\alpha_u = g(y_u)$.
Similarly, for all $y_v < \theta$ and $y_u > y_v$, $v$ is active and $\alpha_v = g(y_v)$.

Again, we decrease $y_v$ gradually from $\theta$ to $0$ and study $M(1,y_v)$.
Observe that at time $y_v$, $u$ is an unmatched neighbor of $v$.
Then there exists a transition time $\tau$ such that $u$ is the favourite neighbor of $v$ when $y_v\in (\tau,\theta)$; and $v$ matches a vertex other than $u$ when $y_v < \tau$.
In summary, we have the following basic gains (refer to Figure~\ref{fig:bip2})
\begin{compactitem}
	\item $\alpha_u = 1-g(\theta)$ when $y_u,y_v > \theta$;
	\item $\alpha_u + \alpha_v = 1$ when $y_v \in (\tau, \theta)$ and $y_u > y_v$;
	\item $\alpha_u=g(y_u)$ when $y_u <\theta$ and $y_v>y_u$;
	\item $\alpha_v = g(y_v)$ when $y_v < \tau$ and $y_u > y_v$.
\end{compactitem}

Next we retrieve some extra gains. Again, the following holds only for bipartite graphs.

\begin{lemma}[Extra Gain]
	When $y_u < \theta$ and $y_v < \tau$, both $u$ and $v$ are matched in $M(y_u, y_v)$.
	When $y_u > \theta$ and $y_v < \tau$, $\alpha_u \geq 1-g(\theta)$.
\end{lemma}
\begin{proof}
	Consider when $y_u < \theta$ and $y_v=1$. According to the previous discussion, $u$ has two unmatched neighbors $z$ and $v$ at time $y_u$.
	Thus $u$ would still be actively matched even if we remove $v$ from the graph. That is, $u$ is active in $M_{v}(y_u,1)$.
	Then by Corollary~\ref{corollary:insertion_bipartite}, after inserting $v$ at any rank, $u$ remains matched. In other words, $u$ is matched in $M(y_u,y_v)$ for all $y_u < \theta$.
	
	By definition of $\tau$, $v$ actively matches a vertex other than $u$ in $M(1,y_v)$ for all $y_v < \tau$. Thus $v$ is active in $M_{u}(1,y_v)$, and matched in $M(y_u,y_v)$ for every $y_u$ by Corollary~\ref{corollary:insertion_bipartite}.
	
	Now we consider the second statement, when $y_u > \theta$ and $y_v < \tau$.
	Observe that in $M(y_u,1)$, $z$ matches $u$ at time $\theta$ while at this moment $v$ is unmatched. Removing $v$ does not affect $z$ and $u$, i.e. $z$ actively matches $u$ in $M_v(y_u,1)$. 
	Then by Corollary~\ref{corollary:insertion_bipartite}, inserting $v$ at any rank does not increase the time that $u$ gets matched. Hence in $M(y_u, y_v)$, $u$ is passively matched at time no later than $\theta$, which implies $\alpha_u \geq 1-g(\theta)$ by the monotonicity of $g$.
\end{proof}

Adding these extra gains to the basic gains, we have (refer to Figure~\ref{fig:bip2})
\begin{equation} \label{eq:uw_bip_ubetter}
	\begin{split}
		\expect{y_u,y_v}{\alpha_u + \alpha_v} \geq & \int_0^\theta (1-y_u)g(y_u)dy_u + \int_0^\tau (1-y_v)g(y_v) dy_v + (1-\theta)(1-\theta+\tau)(1-g(\theta)) \\
		& + \int_0^\theta \min\{y_u,\tau\} m(y_u) dy_u + \int_0^\tau y_v m(y_v) dy_v + \frac{1}{2}(2-\tau-\theta)(\theta-\tau).
	\end{split}
\end{equation}

\paragraph{Analysis of Approximation Ratio.}
To complete the analysis, it remains to find a non-decreasing function $g:[0,1]\to[0,1]$ so that the minimum (over all possible values of $\theta,\tau$) of Equations~\eqref{eq:uw_bip_uv} and~\eqref{eq:uw_bip_ubetter} is at least $0.639$.
We prove Theorem~\ref{th:unweighted_bipartite} by fixing $g$ to be a step function and running a factor revealing LP.
See Appendix~\ref{sec:factor_revealing_lp} for a detailed discussion.

\section{Unweighted General Graphs}\label{sec:unweighted_gen}

Since Corollary~\ref{corollary:insertion_bipartite} holds only for bipartite graphs, the extra gains we proved in the previous section cease to hold for general graphs.
It is easy to check that applying the previous analysis while only having the basic gains, we are not able to beat the $0.5$ barrier on the approximation ratio.

The same difficulty arises in the fully online matching problem~\cite{stoc/HKTWZZ18}. The authors bypass it by introducing a novel concept of ``victim''.
They call a vertex $v$ the victim of $w$ in $M(\vecy)$ if (1) $v$ is a neighbor of $w$; (2) $w$ is active and $v$ is unmatched; (3) $v$ is matched in $M_{w}(\vecy)$.
Intuitively, $v$ is unmatched in $M(\vecy)$ because of the existence of $w$.
It is then shown that either $u,v$ are both matched for some recipe of $y_u,y_v$, or $v$ is the victim of some vertex and receives compensation. In either case, the improved analysis beats the $0.5$ barrier.

In this paper, we introduce a new notion of victim and compensation, which is arguably clearer and more fundamental than the notion given in~\cite{stoc/HKTWZZ18}.
Fix a maximum matching $M^*$, we call $u$ and $v$ perfect partners of each other if $(u,v)\in M^*$.

\begin{definition}[Victim]
	\label{def:victim}
	Suppose in $M(\vecy)$, $z$ actively matches $u$ and $v$ is the perfect partner of $u$.
	Then we call $v$ the \emph{victim} of $z$ if $u$ and $v$ match each other in $M_{z}(\vecy)$.
\end{definition}

Intuitively, the existence of $z$ prevents the algorithm from making the correct decision of matching $u,v$ together. Compared to the definition of Huang et al.~\cite{stoc/HKTWZZ18}, we regard $v$ the victim of $z$ even when $v$ is matched in $M(\vecy)$. The same definition will be applied to edge-weighted graphs in Section~\ref{sec:weighted_general}.
Built upon this definition, we define the following gain sharing rule.

\paragraph{Gain Sharing.}
Let $g$ be a non-decreasing function and $h$ be a function that is pointwise smaller than $g$.
Consider the following two-step gain sharing procedure in matching $M(\vecy)$:
\begin{compactitem}
	\item Whenever $u$ actively matches $v$ at time $y_u$, let $\alpha_u=g(y_u)$ and $\alpha_v=1-g(y_u)$.
	\item For each active vertex $z$ that has an \emph{unmatched} victim $v$, decrease $\alpha_z$ and increase $\alpha_v$ by the same amount $h(y_z)$.
\end{compactitem}

\medskip

We refer to the second step of gain sharing as the compensation step, and the amount $h(y_z)$ of gain as the \emph{compensation} sent from $z$ to $v$.
Note that the compensation step does not change $\sum_{u\in V}\alpha_u$, which means that Lemma~\ref{lemma:dual} can still be applied.
It is easy to see that the passive gain of a vertex $u$ is at least $1-g(y_u)$ and the active gain is at least $g(y_u) - h(y_u)$.

\begin{fact} \label{fact:matched_gain_unweighted_gen}
	If $u$ is matched in $M(\vecy)$, then $\alpha_u \ge m(y_u) \eqdef \min\{g(y_u)-h(y_u), 1-g(y_u)\}$. 
\end{fact}

Moreover, if $v$ is the victim of vertex $z$, either $v$ is matched, $\alpha_v \ge m(y_v)$, or $v$ receives compensation from $z$, $\alpha_v \ge h(y_z)$.
For analysis purpose, we choose $g,h$ so that $\min_y \{m(y)\} \ge \max_y \{h(y)\}$, i.e., the gain of a matched vertex is at least the compensation of an unmatched vertex.
To help understanding, one can imagine the compensation to be a very small amount of gain compared with $m(\cdot)$.
Consequently, we have the following.

\begin{fact}
	If $v$ is the victim of vertex $z$ in $M(\vecy)$, then $\alpha_v \ge h(y_z)$.
\end{fact}

Following the same framework as for bipartite graphs, we fix a pair of perfect partners $u,v$, and fix the decision times of all vertices other than $u,v$ arbitrarily.
Let $M(y_u,y_v)$ denote the realized matching when $u,v$ have decision times $y_u$ and $y_v$, respectively.
Again, we consider whether $(u,v)\in M(1,1)$ and proceed differently.

\subsection{Symmetric Case: $(u,v) \in M(1,1)$}

The analysis is similar to the bipartite case.
Let $\theta$ be the transition time such that $u$ actively matches $v$ in $M(y_u, 1)$ when $y_u > \theta$; matches a vertex other than $v$ in $M(y_u,1)$ when $y_u < \theta$.
The transition time $\tau$ of $y_v$ is defined analogously.

Following the same analysis for bipartite graphs, we have (refer to Figure~\ref{fig:uw_ge_symmetric})
\begin{compactitem}
	\item $\alpha_u + \alpha_v = 1$ when $y_u > \theta$ and $y_v > y_u$; $u$ is active when $y_u < \theta$ and $y_v > y_u$;
	\item $\alpha_u + \alpha_v = 1$ when $y_v > \tau$ and $y_u > y_v$; $v$ is active when $y_v < \tau$ and $y_u > y_v$.
\end{compactitem}

Observe that for general graphs, the gain of an active vertex $u$ is no longer $g(y_u)$, but is lower bounded by $g(y_u) - h(y_u)$.
However, if $u,v$ match each other, then the active vertex does not need to send compensations (recall that $u,v$ are perfect partners).

\begin{figure}[H]
	\vspace*{-5pt}
	\centering
	\includegraphics*[width=0.3\textwidth]{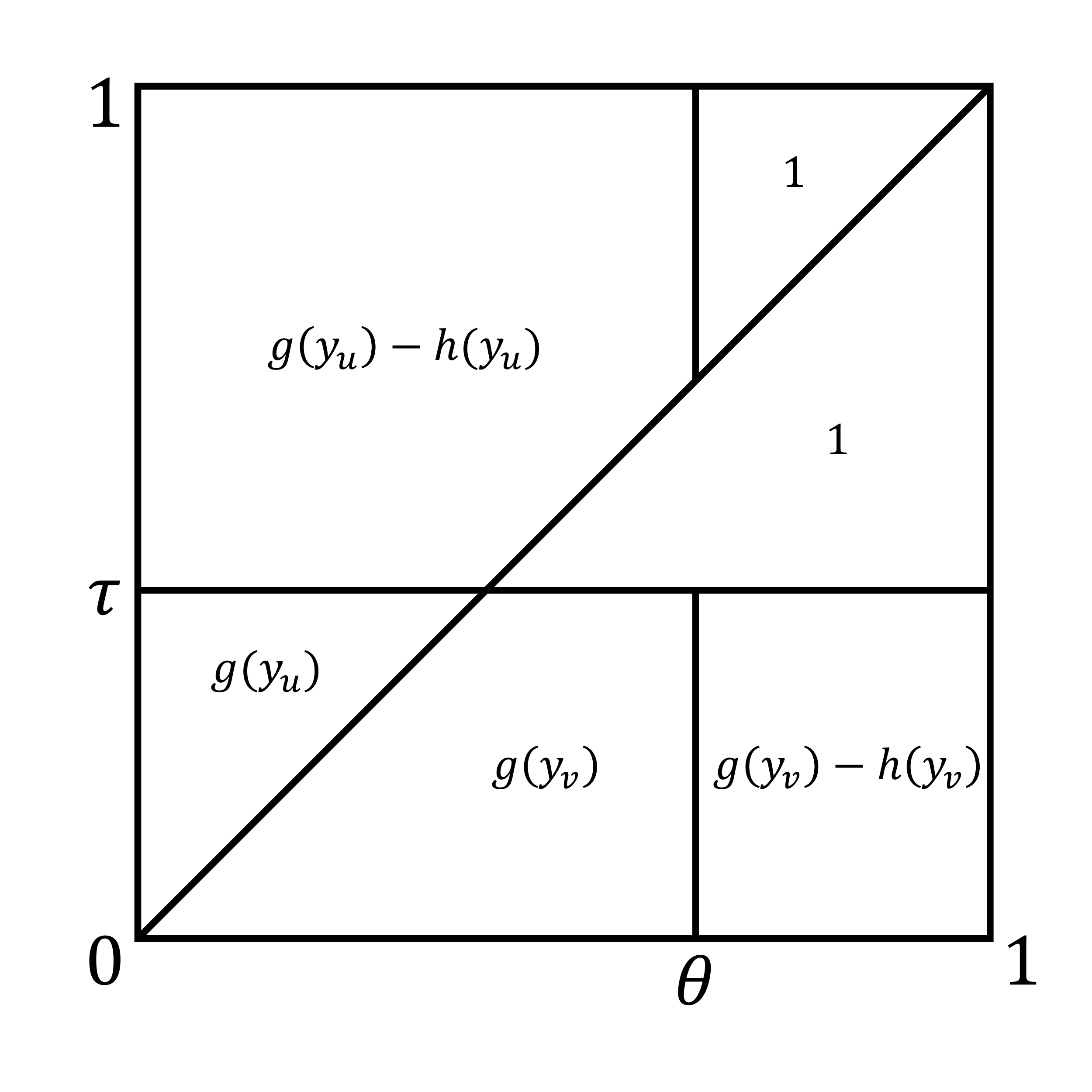}
	\vspace*{-15pt}
	\caption{Unweighted general graphs: $(u,v)\in M(1,1)$.}
	\label{fig:uw_ge_symmetric}
\end{figure}

For bipartite graphs, we show that both $u$ and $v$ are matched in $M(y_u,y_v)$ when $y_u<\theta$ and $y_v < \tau$. Unfortunately, this is not guaranteed in general graphs.
However, we manage to achieve a weaker version of the extra gains that if only one of $u,v$ is matched when $y_u<\theta$ and $y_v <\tau$, then it need not send compensation.

\begin{lemma}[Extra Gain] \label{lemma:unweighted_gen_no_compen}
	For all $y_u < \theta$ and $y_v < \tau$, we have $\alpha_u + \alpha_v \ge g(y_u)$ in $M(y_u,y_v)$ when $y_v>y_u$ and $\alpha_u + \alpha_v \ge g(y_v)$ when $y_u > y_v$.
\end{lemma}
\begin{proof}
	We first consider the case when $y_v > y_u$.
	If $v$ is matched, then $\alpha_u + \alpha_v \geq g(y_u)-h(y_u) + m(y_v) \geq g(y_u)$.
	Now suppose $v$ is unmatched.
	By definition, $v$ actively matches a vertex other than $u$ in $M(1,y_v)$ when $y_v < \tau$. Thus, $v$ is also active in $M_{u}(1,y_v)$.
	In other words, $v$ becomes unmatched after inserting $u$ at decision time $y_u < \theta$.
	We show that in this case $u$ need not send compensation in $M(y_u,y_v)$, which implies $\alpha_u = g(y_u)$.
	
	Suppose $u$ matches $z$ in $M(y_u,y_v)$.
	By Lemma~\ref{lemma:alternating}, removing $u$ triggers an alternating path that starts at $u$ and ends at $v$.
	Thus the perfect partner of $z$ is either matched in both of $M(y_u,y_v)$ and $M_{u}(y_u,y_v)$; or unmatched in both.
	Consequently, $u$ does not have an unmatched victim.
	
	Symmetrically, we have $\alpha_u+\alpha_v \ge g(y_v)$ when $y_u > y_v$.
\end{proof}

In summary, we have the following lower bound on $\expect{}{\alpha_u + \alpha_v}$.
Note that $u,v$ are symmetric. We safely assume that $\tau \leq \theta$ (refer to Figure~\ref{fig:uw_ge_symmetric}).
\begin{equation}
	\label{eq:unweighted_gen_uvmatch}
	\begin{split}
		\expect{y_u,y_v}{\alpha_u + \alpha_v} \ge & \frac{1}{2}(1-\theta)^2 + \frac{1}{2}(1-\tau)^2 
		+ \int_0^\tau \Big( (1-y_v) g(y_v) - (1-\theta) h(y_v) \Big) dy_v \\
		& + \int_{0}^{\theta} \Big( (1-y_u)g(y_u) - (1-\max\{\tau, y_u\})h(y_u) \Big) dy_u.
	\end{split}
\end{equation}

\subsection{Asymmetric Case: $(u,v) \notin M(1,1)$} \label{subsec:unweighted_gen_uearlier}

As before, at least one of $u,v$ is matched before time $1$. We assume without loss of generality that $u$ is matched strictly earlier than $v$ in $M(1,1)$, and let $z$ be the active vertex that matches $u$ with decision time $y_z = \theta$.
First, when $y_u, y_v > \theta$, $u$ is always matched by $z$, and thus $\alpha_u = 1-g(\theta)$.
When $y_u < \theta$, we know that $u$ is active in $M(y_u,1)$, and thus active in $M(y_u,y_v)$ as long as $y_v > y_u$ (by Corollary~\ref{corollary:limit_of_affect}).
Now consider $M(1,y_v)$ when $y_v < \theta$.
Following the same analysis as for bipartite case, let $\tau < \theta$ be the transition time such that $v$ chooses $u$ when $y_v\in (\tau,\theta)$ and chooses a vertex other than $u$ when $y_v\in (0,\tau)$.
Moreover, since $v$ is active in $M_{u}(1,y_v)$ when $y_v < \tau$, following the same analysis as in Lemma~\ref{lemma:unweighted_gen_no_compen}, it can be shown that $\alpha_u + \alpha_v \geq g(y_u)$ when $y_v < \tau$ and $y_u < y_v$.
Similarly, it is easy to show that $\alpha_u + \alpha_v \geq g(y_v)$ when $y_v < \tau$ and $y_u > y_v$\footnote{The key observation is, $u$ is matched in $M_{v}(1,y_v)$, and thus in every $M_{v}(y_u,y_v)$. This implies that after inserting $v$ with $y_v < \tau$, either $u$ is matched, or $v$ need not send compensation.}.

In summary, we have (refer to Figure~\ref{fig:uw_gen_basic})
\begin{compactitem}
	\item[(L1)] $\alpha_u = 1-g(\theta)$ when $y_u > \theta$ and $y_v > \theta$;
	\item[(L2)] $\alpha_u + \alpha_v = 1$ when $y_v \in (\tau,\theta)$ and $y_u > y_v$;
	\item[(L3)] $\alpha_u \geq g(y_u) - h(y_u)$ when $y_u < \min\{\theta, y_v\}$ and $y_v > \tau$;
	\item[(L4)] $\alpha_u + \alpha_v \geq g(y_u)$ when $y_u < y_v$ and $y_v < \tau$;
	\item[(L5)] $\alpha_u + \alpha_v \geq g(y_v)$ when $y_v < \tau$ and $y_u > y_v$.
\end{compactitem}

\vspace*{-10pt}
\begin{figure}[H]
	\centering
	\begin{subfigure}{.3\textwidth}
		\centering
		\includegraphics[width=0.9\linewidth]{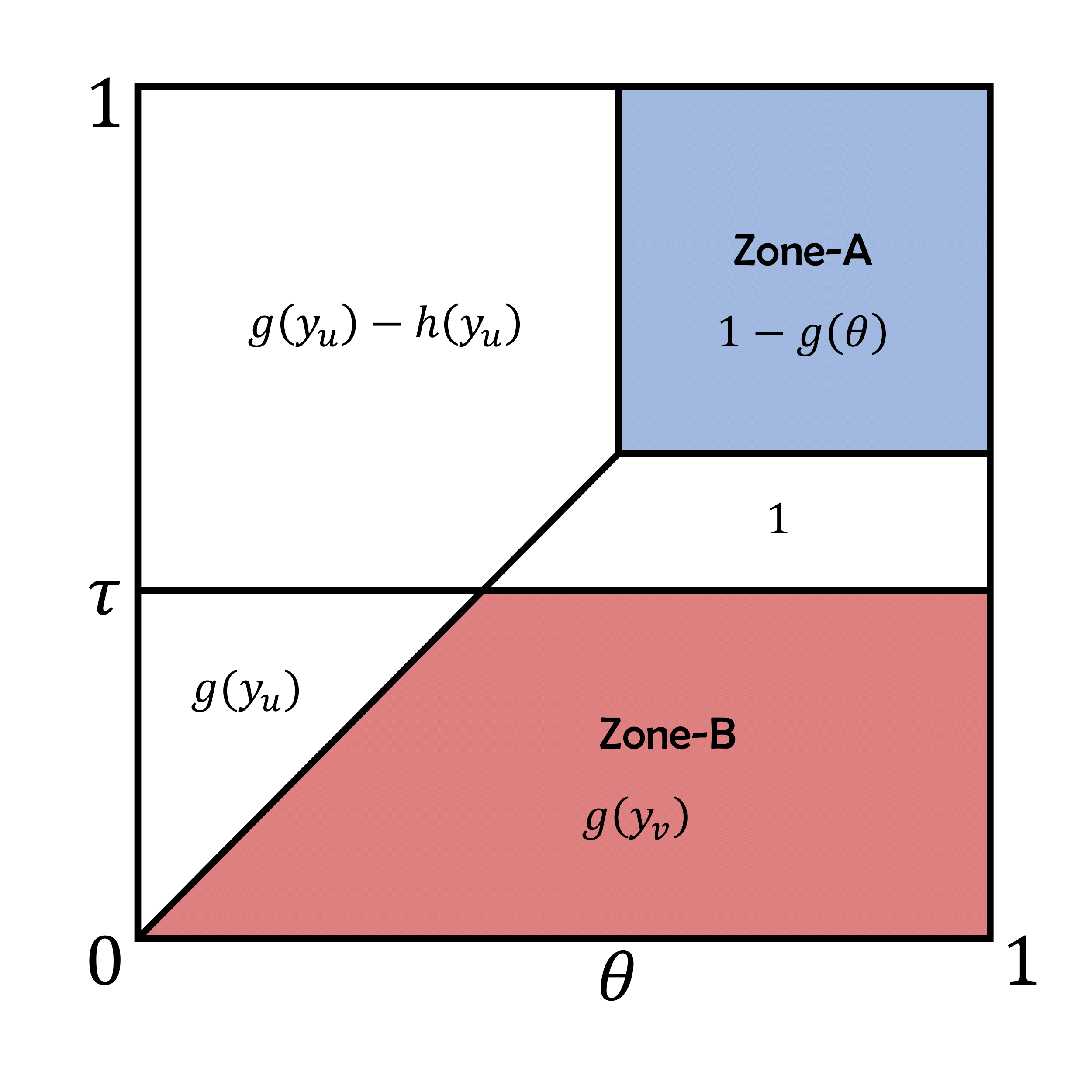}
		\vspace*{-12pt}
		\caption{basic gains}
		\label{fig:uw_gen_basic}
	\end{subfigure}%
	\begin{subfigure}{.3\textwidth}
		\centering
		\includegraphics[width=0.9\linewidth]{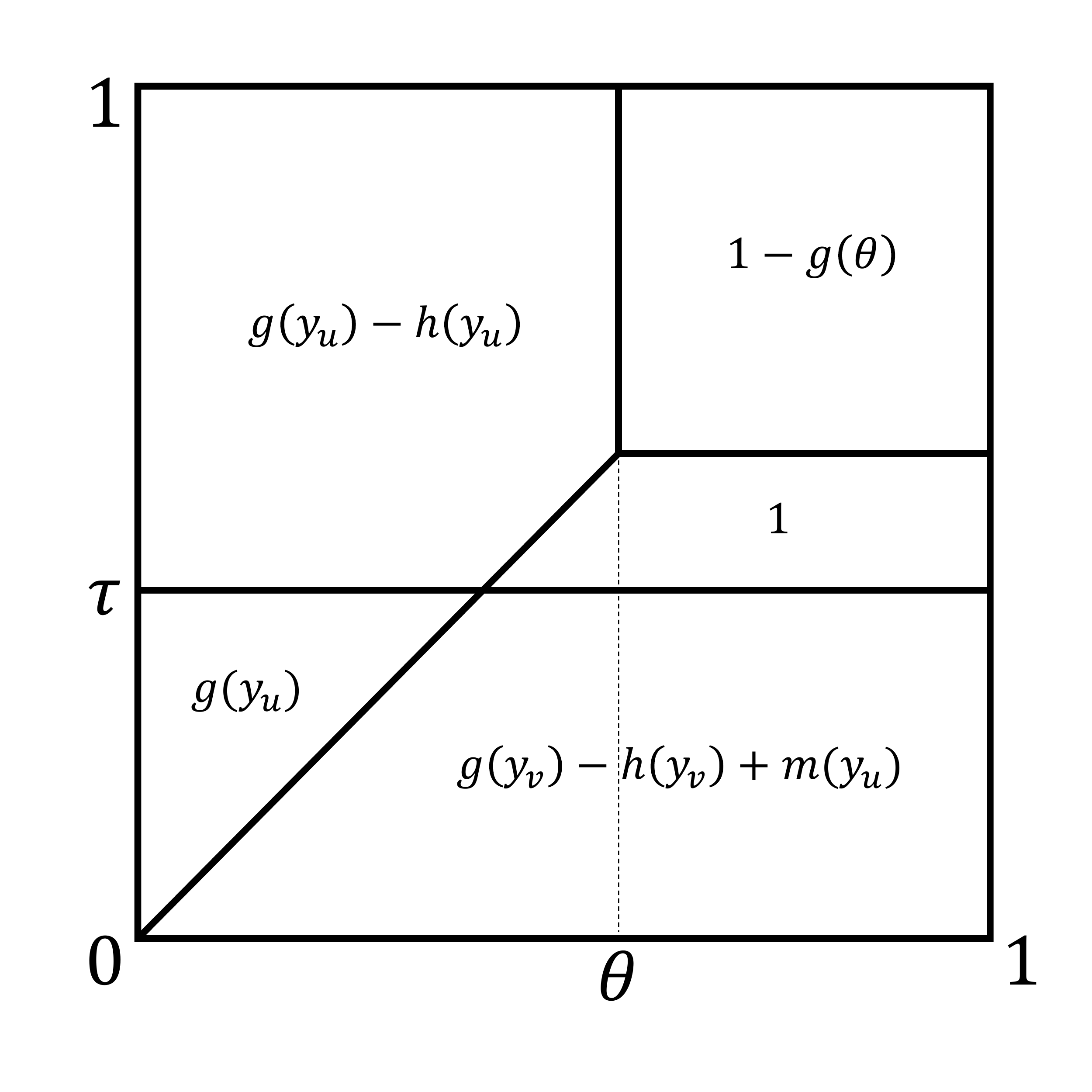}
		\vspace*{-12pt}
		\caption{$u$ is matched earlier}
		\label{fig:uw_gen_uearlier_wo_theta}
	\end{subfigure}%
	\begin{subfigure}{.3\textwidth}
		\centering
		\includegraphics[width=0.9\linewidth]{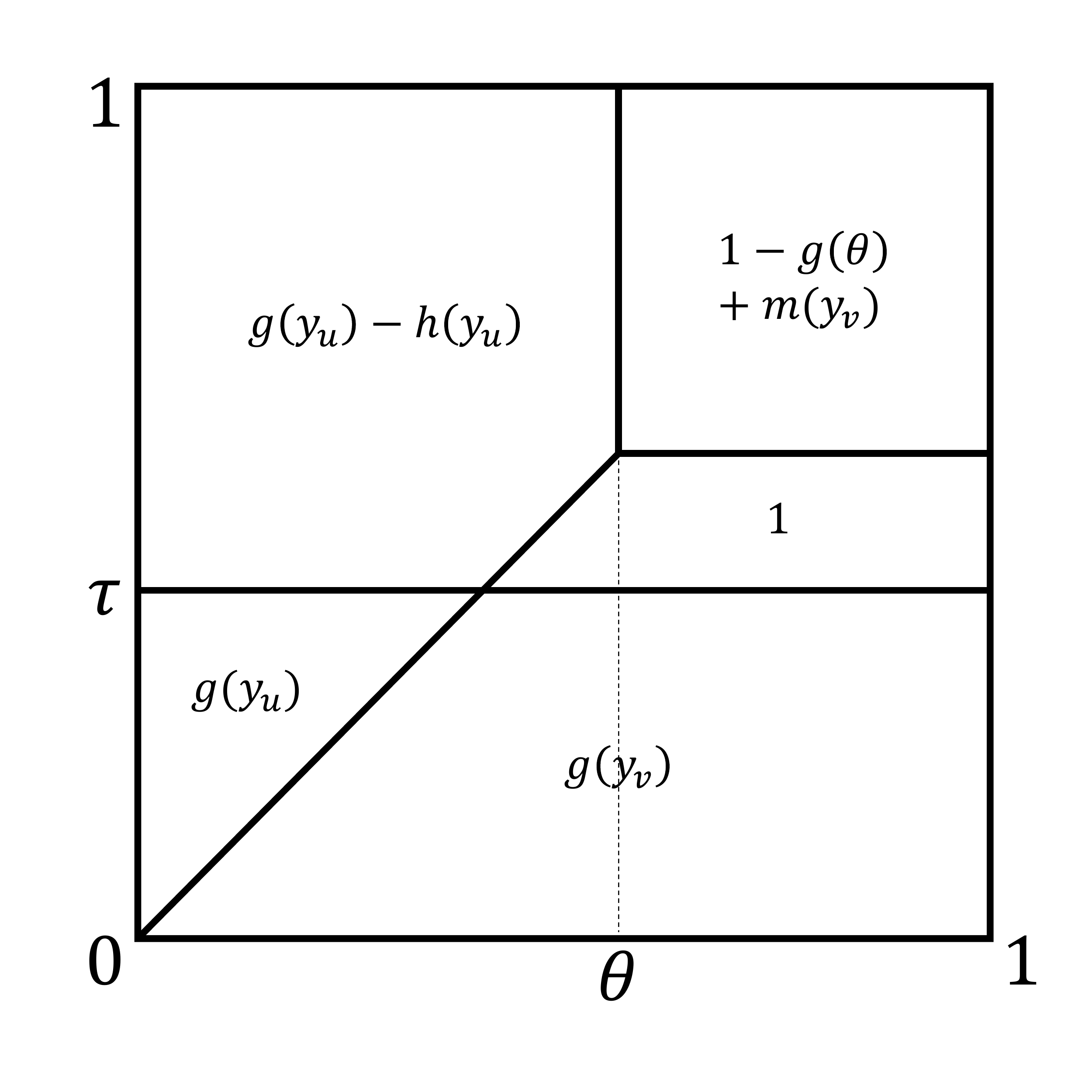}
		\vspace*{-12pt}
		\caption{$v$ is matched earlier}
		\label{fig:uw_gen_vearlier_wo_theta}
	\end{subfigure}%
	\caption{Simple lower bounds and the case when $(u,v)\notin M_{z}(1,1)$.}
\end{figure}
Unsurprisingly, the above basic gains do not yield an approximation ratio strictly above $0.5$.

\paragraph{Extra Gains.}
For convenience of discussion, we define \textsf{Zone-A} to be the matchings $M(y_u,y_v)$ when $y_u>\theta$ and $y_v>\theta$ (where lower bound (L1) is applied), and \textsf{Zone-B} to be the matchings $M(y_u,y_v)$ when $y_v<\tau, y_u>y_v$ (where lower bound (L5) is applied).
In the following, we show that better lower bounds can be obtained for either (L1) or (L5).
Roughly speaking, if $v$ is unmatched in \textsf{Zone-A}, then either it is compensated by $z$ (in which case (L1) can be improved), or $u$ will be matched in \textsf{Zone-B} (in which case (L5) can be improved).
Hence depending on the matching status of $u$ and $v$ in $M_{z}(1,1)$, i.e., when $z$ is removed, we divide our analysis into two cases.

\subsubsection{Case 1: $(u,v)\notin M_{z}(1,1)$}

In this case, at least one of $u,v$ is matched passively before time $1$ in $M_{z}(1,1)$.
We first consider the case when $u$ is matched strictly earlier than $v$.
We show that in this case $u$ is matched in \textsf{Zone-B}, and thus (L5) can be improved (see Figure~\ref{fig:uw_gen_uearlier_wo_theta}).

\begin{lemma}\label{lemma:uw_ge_uearlier}
	If $u$ is matched strictly earlier than $v$ in $M_{z}(1,1)$, then $u$ is matched in \textsf{Zone-B}.
\end{lemma}
\begin{proof}
	To show that $u$ is matched in \textsf{Zone-B}, by Lemma~\ref{lemma:monotonicity} it suffices to show that $u$ is matched in $M(1,y_v)$, for all $y_v < \tau$.
	Suppose otherwise, i.e., $u$ is unmatched in $M(1,y_v)$ for some $y_v < \tau$.
	
	Since $z$ matches $u$ in $M_{v}(1,y_v)$, the symmetric difference between $M(1,y_v)$ and $M_{v}(1,y_v)$ is an alternating path that starts from $v$ and ends with $u$. Moreover, $z$ is the second last vertex in the alternating path.
	Now suppose we remove $v$ and $z$ simultaneously in $M(1,y_v)$.
	Then in the resulting matching, all vertices between $v$ and $z$ in the alternating path recover their matching status in $M_{v}(1,y_v)$, while all other vertices remain the same matching status.
	In particular, $u$ remains unmatched when both $v$ and $z$ are removed from $M(1,y_v)$.
	However, since $u$ is matched strictly earlier than $v$ in $M_{z}(1,1)$, $u$ should remain matched to the same vertex when we further remove $v$, which is a contradiction.
\end{proof}

Given the lemma, we improve (L5) and obtain the following (refer to Figure~\ref{fig:uw_gen_uearlier_wo_theta}).
As we will show later, this is not the bottleneck case since this lower bound is strictly larger than \eqref{eq:unweighted_gen_gamma>theta}.
\begin{equation} \label{eq:unweighted_gen_uearlier}
	\begin{split}
		\expect{y_u,y_v}{\alpha_u + \alpha_v} \ge & (1-\theta)^2(1-g(\theta)) + \frac{1}{2}(2-\theta-\tau)(\theta-\tau) + \int_{0}^{\tau} (1-y_v)(g(y_v) - h(y_v)) dy_v\\
		+ \int_{0}^{\theta}& \Big( (1-y_u)g(y_u) - (1-\max\{\tau, y_u\})h(y_u) \Big) dy_u + \int_{0}^{1} \min \{\tau, y_u\} m(y_u) dy_u.
	\end{split}
\end{equation}

Next, we consider the case when $v$ is matched strictly earlier than $u$ in $M_{z}(1,1)$.
We show that in this case $v$ is matched in \textsf{Zone-A}, which improves (L1).

\begin{lemma}\label{lemma:uw_ge_vearlier}
	If $v$ is matched earlier than $u$ in $M_{z}(1,1)$, then $v$ is matched in \textsf{Zone-A}.
\end{lemma}
\begin{proof}
	Consider adding $z$ back to $M_z(1,1)$. Since $z$ chooses $u$, $v$ remains matched in $M(1,1)$. Moreover,  all matchings $M(y_u,1)$ for $y_u>\theta$ are the same and hence, $v$ is matched in $M(y_u,1)$ for all $y_u > \theta$. Finally, by Lemma~\ref{lemma:monotonicity} $v$ is matched in \textsf{Zone-A}.
\end{proof}

Thus, we obtain the following lower bound (refer to Figure~\ref{fig:uw_gen_vearlier_wo_theta}).
As we will show later, this is neither the bottleneck case since the lower bound is strictly larger than \eqref{eq:unweighted_gen_gamma=theta}.
\begin{equation}
	\begin{split}
		\expect{y_u,y_v}{\alpha_u + \alpha_v} \ge & (1-\theta)^2 (1-g(\theta)) + (1-\theta) \int_\theta^1 m(y_v) dy_v +  \frac{1}{2}(2-\theta-\tau)(\theta-\tau)  \\
		+ \int_0^\theta & (1-y_v)g(y_v) dy_v + \int_0^\theta \Big( (1-y_u)\cdot g(y_u) - (1-\max\{\tau, y_u\})h(y_u) \Big) dy_u.
	\end{split}
	\label{eq:unweighted_gen_vearlier}
\end{equation}

\subsubsection{Case 2: $(u,v)\in M_{z}(1,1)$}

Now we study the second case when $u,v$ match each other in $M_{z}(1,1)$.
This is where the notion of victim applies. Formally, we have the following lower bound of $\alpha_v$ in \textsf{Zone-A}.

\begin{lemma}[Compensation] \label{lemma:compensation_unweighted}
	For all $y_u,y_v>\theta$, if $u$ matches $v$ in $M_{z}(y_u,1)$, then we have $\alpha_v \geq h(\theta)$ in $M(y_u,y_v)$.
\end{lemma}
\begin{proof}
	When $y_u < y_v$, $u$ actively matches $v$ in $M_z(y_u,y_v)$. Hence $v$ is the victim of $z$ and $\alpha_v \ge h(\theta)$. 
	When $y_v < y_u$, either $v$ actively matches $u$ in $M_z(y_u,y_v)$ and thus $v$ is the victim of $z$ in $M(y_u,y_v)$, or $v$ actively matches a vertex other than $u$. In the first case, $\alpha_v \ge h(\theta)$. In the second case, when we add back $z$ to the graph, $z$ chooses $u$ and does not affect the matching status of $v$. Hence, $\alpha_v \ge m(y_v) \ge h(\theta)$.
\end{proof}

To apply this lemma, we first consider the case when $u$ matches $v$ in all $M_{z}(y_u,1)$, where $y_u > \theta$. 
The lemma implies that $\alpha_v \geq h(\theta)$ in \textsf{Zone-A}. Refer to Figure~\ref{fig:uw_ge_compensation}, we have
\begin{equation}
	\label{eq:unweighted_gen_gamma=theta}
	\begin{split}
		\expect{y_u,y_v}{\alpha_u + \alpha_v} \ge & (1-\theta)^2 (1-g(\theta)+h(\theta)) + \frac{1}{2}(2-\tau-\theta)(\theta-\tau) + \int_0^\tau (1-y_v)g(y_v) dy_v\\
		& + \int_{0}^{\theta} \Big( (1-y_u)g(y_u) - (1-\max\{\tau, y_u\})h(y_u) \Big) dy_u.
	\end{split}
\end{equation}

\vspace*{-15pt}
\begin{figure}[H]
	\centering
	\begin{subfigure}{.4\textwidth}
		\centering
		\includegraphics[width=0.75\linewidth]{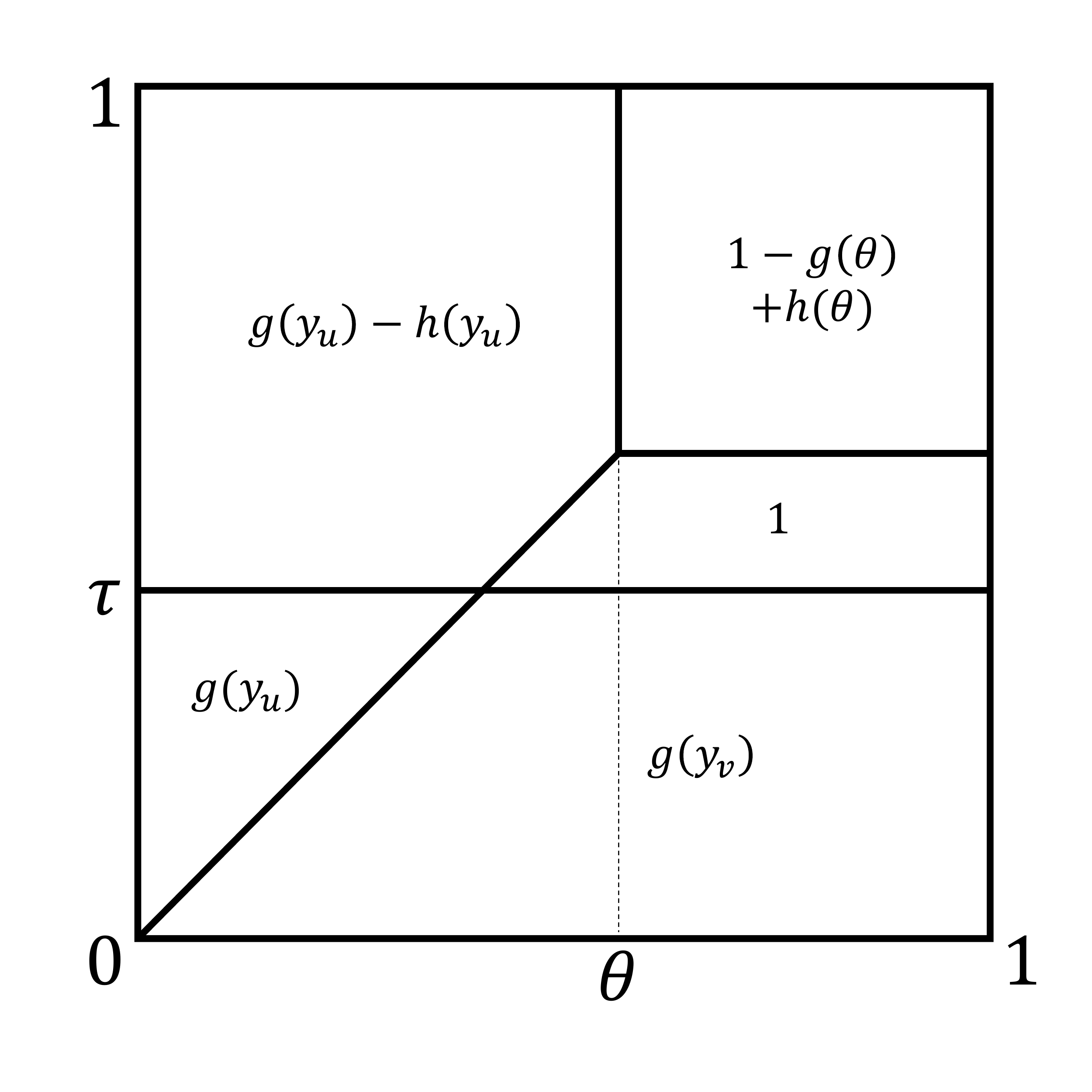}
		\vspace*{-10pt}
		\caption{$u$ matches $v$ in $M_{z}(y_u,1)$ when $y_u > \theta$}
		\label{fig:uw_ge_compensation}
	\end{subfigure}
	\hspace*{10pt}
	\begin{subfigure}{.4\textwidth}
		\centering
		\includegraphics[width=0.75\linewidth]{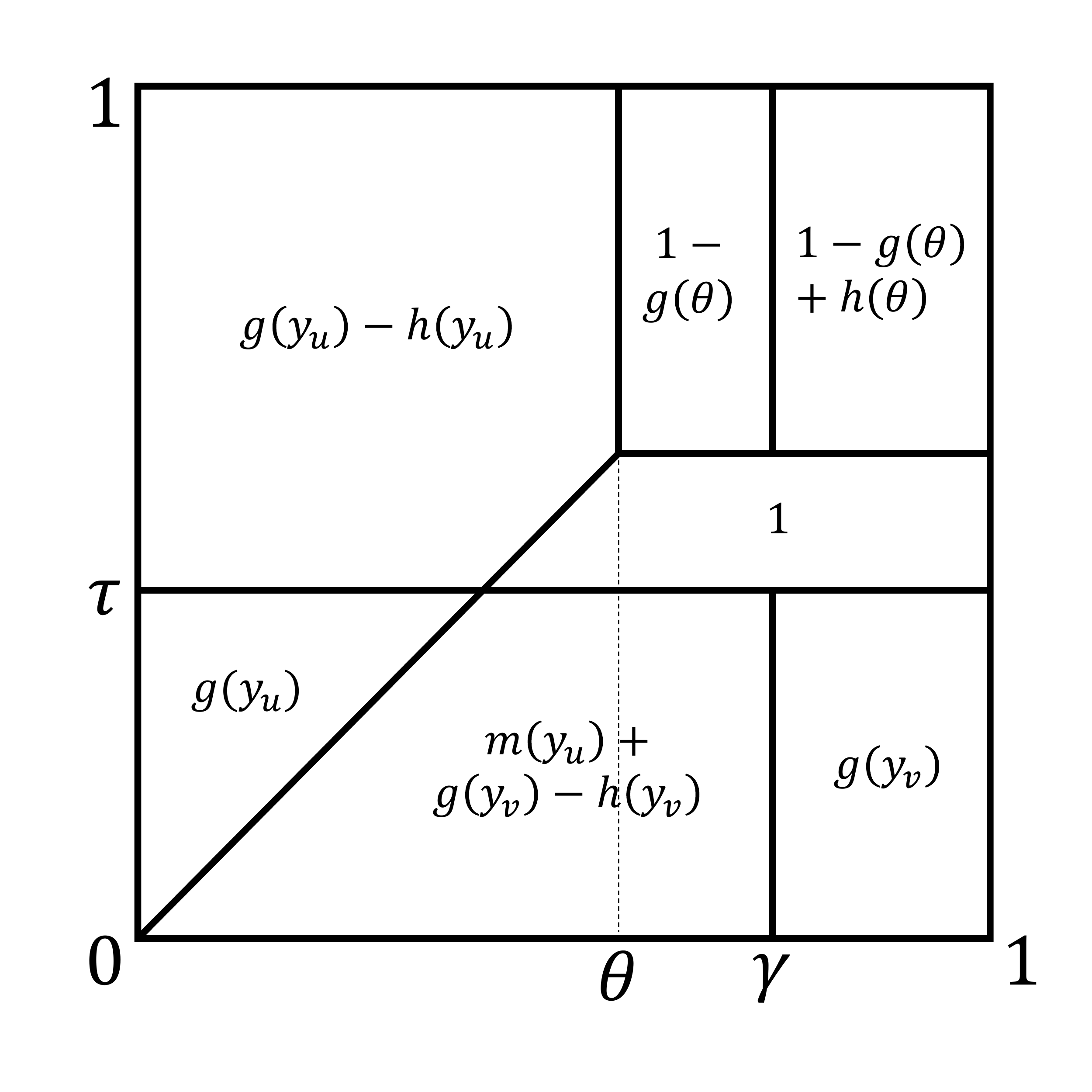}
		\vspace*{-10pt}
		\caption{$u$ matches $v$ in $M_{z}(y_u,1)$ when $y_u > \gamma$}
		\label{fig:uw_ge_gamma}
	\end{subfigure}
	\caption{$u,v$ match each other in $M_{z}(1,1)$.}
\end{figure}

Since $h(x) \leq m(y)$ for all $x,y\in [0,1]$, it is obvious that this lower bound is no larger than~\eqref{eq:unweighted_gen_vearlier}.
We would like to remark that in this case, we can see clearly how the compensation rule helps us achieve an approximation ratio strictly above $0.5$.
Without the compensation $v$ receives in \textsf{Zone-A}, the lower bounds become Figure~\ref{fig:uw_gen_basic}, which cannot beat $0.5$. 

\medskip

Finally, we consider the case when $u$ does not always match $v$ in $M_z(y_u,1)$ for all $y_u > \theta$.
Let $\gamma > \theta$ be the transition time such that $u$ matches $v$ in  $M_z(y_u,1)$ when $y_u >\gamma$ and matches a vertex other than $v$ when $y_u \in (\theta,\gamma)$ (see Figure~\ref{fig:uw_ge_gamma}).

Applying Lemma~\ref{lemma:compensation_unweighted} to $M(y_u,y_v)$ where $y_u > \gamma$ and $y_v > \theta$, we have $\alpha_v \ge h(\theta)$.
Now consider the matching $M_{z}(y_u,1)$, where $y_u < \gamma$.
Intuitively, since $u$ is matched strictly earlier than $v$, (in the same spirit of Lemma~\ref{lemma:uw_ge_uearlier}) $u$ has another neighbor as a backup, and hence is still matched when we decrease $y_v$. Formally, we have the following.
\begin{lemma}
	When $y_u < \gamma$ and $y_v < \tau$, $u$ is matched in $M(y_u,y_v)$.
\end{lemma}
\begin{proof}
	The proof is similar to that of Lemma~\ref{lemma:uw_ge_uearlier}. It suffices to show that $u$ is matched in $M(\gamma,y_v)$ for all $y_v < \tau$.
	By definition of $\gamma$, $u$ actively matches a vertex other than $v$ in $M_{z}(\gamma,1)$.
	Thus $u$ remains active if we further remove $v$ from the graph.
	
	On the other hand, if $u$ is unmatched in $M(\gamma,y_v)$, then the symmetric difference between $M(\gamma,y_v)$ and $M_{v}(\gamma,y_v)$ is an alternating path that starts from $v$ and ends $u$. Moreover, $z$ is the second last vertex in the alternating path.
	Consequently, $u$ remains unmatched if we remove both $v$ and $z$ from the graph, which is a contradiction.
\end{proof}

Plugging in the improved version of (L1) and (L5), we obtain the following (refer to Figure~\ref{fig:uw_ge_gamma}).
\begin{equation} \label{eq:unweighted_gen_gamma>theta}
	\begin{split}
		\expect{y_u,y_v}{\alpha_u + \alpha_v} \ge & (1-\theta)^2(1-g(\theta)) + (1-\gamma)(1-\theta)h(\theta) + \frac{1}{2}(2-\tau-\theta)(\theta-\tau) \\
		+ & \int_{0}^{\theta} \Big( (1-y_u)g(y_u) - (1-\max\{\tau, y_u\})h(y_u) \Big) dy_u + \int_0^\gamma \min\{y_u,\tau\} m(y_u) dy_u \\
		+ & \int_0^\tau \Big( (1-y_v)g(y_v) - (\gamma-y_v)h(y_v) \Big) dy_v.
	\end{split}
\end{equation}

Observe that when $\gamma = 1$, this bound degenerates to~\eqref{eq:unweighted_gen_uearlier}.
It is easy to see this by comparing Figure~\ref{fig:uw_ge_gamma} and Figure~\ref{fig:uw_gen_uearlier_wo_theta}.

\paragraph{Analysis of Approximation Ratio.}
Equipped with the previous lower bounds, it suffices to find functions $g,h:[0,1] \to [0,1]$ so that the minimum of \eqref{eq:unweighted_gen_uvmatch}, \eqref{eq:unweighted_gen_gamma=theta}, \eqref{eq:unweighted_gen_gamma>theta} over all possible $\theta,\tau,\gamma$ is maximized. 
Again, a factor revealing LP shows that \rdt is $0.531$-approximate, finishing the proof of Theorem~\ref{th:unweighted_general}.

\section{Weighted General Graph} \label{sec:weighted_general}

We analyze the approximation ratio of our algorithm on weighted general graphs in this section.
Recall that our algorithm probes pairs $(u,v)$ in descending order of the perturbed weights $(1-g(\min\{y_u,y_v\}))w_{uv}$.
Sometimes it will be helpful to interpret the algorithm as replacing every (potential) edge $(u,v)$ with two directed edges $(u,v)$ and $(v,u)$.
Then we set the perturbed weight of a directed edge $(u,v)$ to be $(1-g(y_u))w_{uv}$ and probe the directed edges in descending order of their perturbed weights. 

Our analysis is structured similarly to the unweighted case.
However, we will see many of the previous properties fail when the analysis goes into details.
We first provide some basic properties of \textsf{Perturbed Greedy} for weighted graphs, which will be the building blocks of our analysis.
Notably, we generalize the gain sharing and compensation rule to edge-weighted graphs.

First we observe the following property that is analogous to Lemma~\ref{lemma:alternating} for the unweighted case, where we substitute \emph{decision times} with \emph{perturbed weights}.
The proof is almost identical to that of Lemma~\ref{lemma:alternating}, thus is omitted.

\begin{lemma}[Weighted Alternating Path]\label{lemma:weighted_alternating}
	If $u$ is matched in $M(\vecy)$, the symmetric difference between $M(\vecy)$ and $M_{u}(\vecy)$ is an alternating path $(u,u_1,u_2,\ldots)$ in which the perturbed weights of edges are decreasing.
	Consequently, vertices $u_1,u_3,\ldots$ are matched earlier\footnote{There is no explicit concept of time. However, since the edges are probed in descending order of their perturbed weights, a vertex being matched earlier means that the new edge has larger perturbed weight than the old one.} in $M(\vecy)$ than in $M_{u}(\vecy)$.
\end{lemma}

Following the same definition of victim, we define the gain sharing rules for edge-weighted graphs.
For technical reasons, we set $h(y) = \frac{1}{10}(1-g(y))$ and restrict ourselves to non-decreasing function $g$ such that $g(y)\in [0.4,0.6]$ for all $y\in[0,1]$.
As a consequence, for all $y\in [0,1]$ we have
\begin{equation*}
	m(y) \eqdef \min\{g(y)-h(y), 1-g(y)\} = \min\{1.1 g(y)- 0.1, 1-g(y)\} \ge 0.34 > 5 \cdot \max_x\{h(x)\}.
\end{equation*}

\paragraph{Gain Sharing.} 
Consider the following two-step approach for gain sharing in matching $M(\vecy)$:
\begin{compactitem}
	\item Whenever an edge $(u,v)$ is added to the matching with $u$ active and $v$ passive, let $\alpha_u = g(y_u)\cdot w_{uv}$ and $\alpha_v =  (1-g(y_u))\cdot w_{uv}$.
	\item For each active vertex $z$, if $z$ has a victim $v$, decrease $\alpha_z$ and increase $\alpha_v$ by the same amount such that $\alpha_v \geq h(y_z)\cdot w_{uz}$ afterwards, where $u$ is the vertex matched by $z$.
	More specifically, the amount of compensation is (where $[t]^+ \eqdef \max\{t,0\}$)
	\begin{compactitem}
		\item $h(y_z)\cdot w_{uz}$ if $v$ is unmatched;
		\item at most $\left[h(y_z)\cdot w_{uz} - (g(y_v)-h(y_v))\cdot w_{vx}\right]^+$ if $v$ actively matches some vertex $x$;
		\item $\left[h(y_z)\cdot w_{uz} - (1 - g(y_x))\cdot w_{vx}\right]^+$ if $v$ is passively matched by $x$.
	\end{compactitem}
\end{compactitem}

\medskip

Note that the above compensation step is consistent with the unweighted case: if the victim $v$ is matched, then the amount of compensation is $0$, since $m(y_v) > h(y_z)$; otherwise it is $h(y_z)$.

\paragraph{Lower Bounds of Gains.}
Suppose $u$ is matched with $v$.
By the gain sharing rules, if $u$ is active, then its gain is at least $(g(y_u)- h(y_u))w_{uv}$; if it is passive, then its gain is $(1-g(y_u))w_{uv}$.
Thus the gain of a matched vertex $u$ is lower bounded by $m(y_u) w_{uv}$.
While the amount of compensation a victim $v$ (of $z$) receives depends on the gain of $v$ in the first step, our compensation rule guarantees $\alpha_v \geq h(y_z)\cdot w_{zu}$ afterwards, where $u$ is the perfect partner of $v$ and is matched by $z$ in $M(\vecy)$.

\medskip

We follow the previous framework by fixing an arbitrary pair of perfect partners $(u,v)$, and fixing the ranks of all vertices other than $u,v$ arbitrarily.
We derive lower bounds on $\alpha_u + \alpha_v$ for every $M(y_u,y_v)$, and show that the integration (over $y_u$ and $y_v$) of the lower bound is at least $0.5014$.
For convenience, we assume $w_{uv} = 1$.

In the remaining part of this section, we say that a vertex $u$ is matched \emph{strictly earlier} than $v$ in matching $M(\vecy)$ if $v$ remains unmatched after $u$ is matched; we say that $u$ is matched \emph{earlier} than $v$ if either $u$ is strictly earlier than $v$, or $u$ actively matches $v$.

\begin{fact} \label{fact:matched_ealier}
	Suppose $u$ is matched earlier than $v$ in $M(y_u,y_v)$.
	\begin{compactitem}
		\item[1.] Increasing the rank of $v$ does not change the matching status of $u$.
		\item[2.] If $u$ is active, $\alpha_u \ge g(y_u)-h(y_u)$; if $u$ is passive, $\alpha_u \ge 1-g(y_u)$.
	\end{compactitem}
\end{fact}
\begin{proof}
	For the first statement, increasing the rank of $v$ can not increase the perturbed weights of edges adjacent to $v$.
	Thus before $u$ is matched, all probes have the same results before and after the increment of $y_v$, which means $u$ is matched to the same neighbor.
	
	For the second statement, suppose $u$ is matched with some vertex $z$ (which can be $v$).
	If $u$ is active, then $(1-g(y_u)) w_{uz} \ge (1-g(y_u)) w_{uv}$ since edge $(u,z)$ is probed no later than $(u,v)$.
	Hence $\alpha_u \geq (g(y_u)-h(y_u)) w_{uz} \geq g(y_u)-h(y_u)$.
	When $u$ is passive, the perturbed weight of $(u,z)$ equals $(1-g(y_z))w_{uz}$, which is at least the perturbed weight of $(u,v)$. Thus we have $\alpha_u \ge (1-g(y_z))w_{uz} \ge 1-g(y_u)$.
\end{proof}

\begin{lemma} \label{lemma:no-compen}
	Suppose $u$ is active and matched earlier than $v$ in $M(y_u,y_v)$.
	If $v$ is matched with some $x$ such that $w_{vx}\geq \frac{1}{2}$ in $M_{u}(y_u,y_v)$, then we have $\alpha_u+\alpha_v \ge g(y_u)$ in $M(y_u,y_v)$.
\end{lemma}
\begin{proof}
	By Fact~\ref{fact:matched_ealier}, we have $\alpha_u \geq g(y_u)-h(y_u)$.
	If $v$ remains matched with $x$ in $M(y_u,y_v)$, then the lemma holds since $\alpha_u+\alpha_v \geq g(y_u) - h(y_u) + \frac{1}{2} m(y_v) \ge g(y_u)$ (recall that $\min_{y}\{m(y)\} \geq 5\cdot \max_{x}\{h(x)\}$.).
	If $u$ does not have a victim, or only need to send $0$ compensation to its victim, then we are also done.
	
	Otherwise, by Lemma~\ref{lemma:weighted_alternating}, the symmetric different between $M(y_u,y_v)$ and $M_{u}(y_u,y_v)$ is an alternating path starting from $u$ that contains $(v,x)$.
	Let $z$ be matched by $u$ in $M(y_u,y_v)$ and $z^*$ be the victim of $u$.
	Since the edge $z^*$ matches in $M(y_u,y_v)$ appears no later than $(v,x)$ in the alternating path, by Lemma~\ref{lemma:weighted_alternating}, the perturbed weight of this edge is at least $(1-g(\min\{y_v,y_x\}))\frac{1}{2}\geq 0.2$ (recall that $g(\cdot)\in [0.4,0.6]$).
	Hence the gain of $z^*$ in $M(y_u,y_v)$ before the compensation step is at least $m(y_{z^*})\cdot \frac{0.2}{1-g(0)} = m(y_{z^*})\cdot \frac{1}{3}$.
	Consequently the gain of $u$ after the compensation step is $\alpha_u \ge (g(y_u)-h(y_u))w_{uz} + m(y_{z^*})\cdot \frac{1}{3}\ge g(y_u)$.
\end{proof}

We regard the above lemma as a weighted generalization of Lemma~\ref{lemma:unweighted_gen_no_compen}, which says when $u$ is active, it need not send compensation if $v$ matches an edge that is not too bad when $u$ is removed.

Now we are ready to derive lower bounds on $\alpha_u + \alpha_v$ for every $M(y_u,y_v)$.
Similar as before, we divide our analysis into two cases depending on whether $(u,v)\in M(1,1)$.

\subsection{Symmetric Case: $(u,v) \in M(1,1)$}

Since $u,v$ are matched together in $M(1,1)$, we define $\theta$ to be the transition rank such that $u$ matches $v$ in $M(y_u,1)$ when $y_u > \theta$; matches a vertex other than $v$ in $M(y_u,1)$ when $y_u\in (0,\theta)$.
We define $\lambda$ analogously for $v$.
Assume w.l.o.g. that $\theta \geq \lambda$.

First, we show that for $y_u > \theta$ and $y_v > \lambda$, $u,v$ are matched together in $M(y_u, y_v)$.
Suppose otherwise, and assume $u$ is matched strictly earlier.
By Fact~\ref{fact:matched_ealier}, if we increase $y_v$ to $1$, the matching status of $u$ should not be affected. 
Hence $u$ is not matched with $v$ in $M(y_u,1)$, which contradicts the definition of $\theta$ (recall that $y_u > \theta$).
The same argument implies that $v$ is not matched strictly earlier. 
Hence $\alpha_u + \alpha_v = 1$ when $y_u > \theta$ and $y_v > \lambda$ (recall that $u,v$ are perfect partners).

Under the same logic, when $y_u < \theta$ and $y_v > \lambda$, $v$ is not matched strictly earlier than $u$.
Otherwise $v$ is also not matched with $u$ in $M(1,y_v)$, which contradicts the definition of $\lambda$.
Hence when $y_u < \theta$ and $y_v > \lambda$, $u$ is active and matched earlier than $v$.
By Lemma~\ref{fact:matched_ealier} we have $\alpha_u + \alpha_v \geq g(y_u)-h(y_u)$\footnote{It is possible that $u$ is active and matched strictly earlier than $v$ even when $y_u > y_v$. This is a key difference between the weighted and unweighted case: smaller rank does not necessarily imply earlier decision time.}.
Symmetrically, we have $\alpha_u + \alpha_v \geq g(y_v)-h(y_v)$ when $y_u > \theta$ and $y_v < \lambda$.

Finally, we consider the case when $y_u < \theta$ and $y_v < \lambda$.
We show that in this case we must have $\alpha_u + \alpha_v \geq \min\{g(y_u), g(y_v)\}$.
Suppose $u$ is matched earlier than $v$.
Then $u$ must be active, as otherwise $u$ is also passive in $M(y_u,1)$.
By definition of $\lambda$, we know that in $M_{u}(y_u,y_v)$, the edge $v$ matches has weight at least $w_{uv}=1$.
Applying Lemma~\ref{lemma:no-compen}, we have $\alpha_u + \alpha_v \geq g(y_u)$.
Similarly, when $v$ is matched earlier than $u$, we have $\alpha_u+\alpha_v \geq g(y_v)$.
Therefore, $\alpha_u + \alpha_v \ge \min\{g(y_u),g(y_v)\}$.
Given that $g$ is non-decreasing, $\alpha_u + \alpha_v \geq g(y_u)$ when $y_u < y_v$ and $\alpha_u + \alpha_v \geq g(y_v)$ when $y_v < y_u$.

In summary, we have (refer to Figure~\ref{fig:weighted_general_uvmatch})
\begin{equation} \label{eq:weighted_general_uvmatch}
	\begin{split}
		\expect{y_u,y_v}{\alpha_u + \alpha_v} \ge & (1-\lambda) \int_0^\theta \Big( g(y_u)-h(y_u) \Big)dy_u + (1-\theta) \int_0^\lambda \Big( g(y_v)-h(y_v) \Big) dy_v \\
		& + \int_0^\lambda (\lambda-y_u)g(y_u) dy_u + \int_0^\lambda (\theta-y_v)g(y_v) dy_v + (1-\theta)(1-\lambda).
	\end{split}
\end{equation}

\vspace*{-10pt}
\begin{figure}[H]
	\centering
	\begin{subfigure}[H]{.3\textwidth}
		\centering
		\includegraphics*[width=0.9\textwidth]{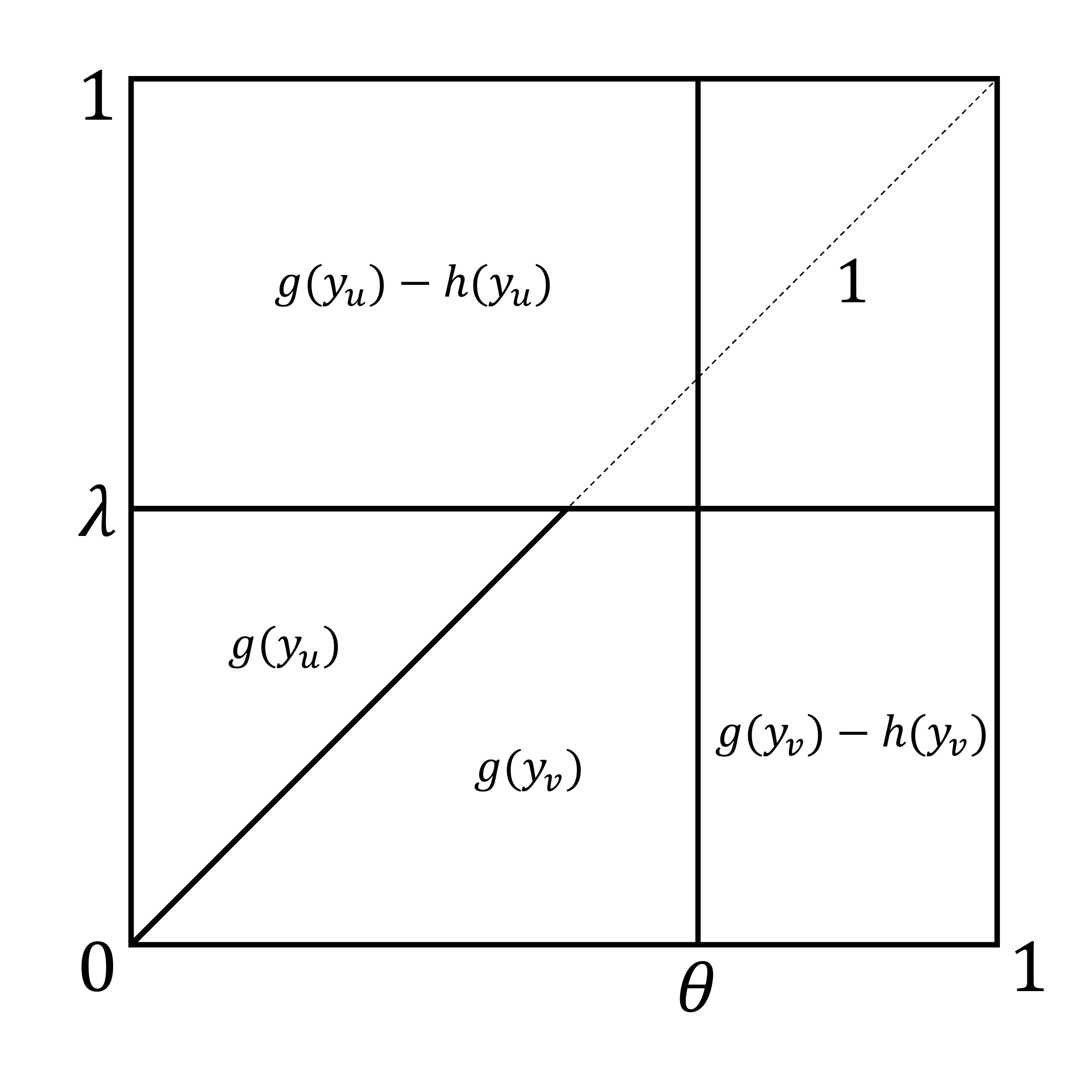}
		\vspace*{-10pt}
		\caption{$(u,v) \in M(1,1)$.}
		\label{fig:weighted_general_uvmatch}
	\end{subfigure}
	\begin{subfigure}[H]{.3\textwidth}
		\centering
		\includegraphics[width=0.9\linewidth]{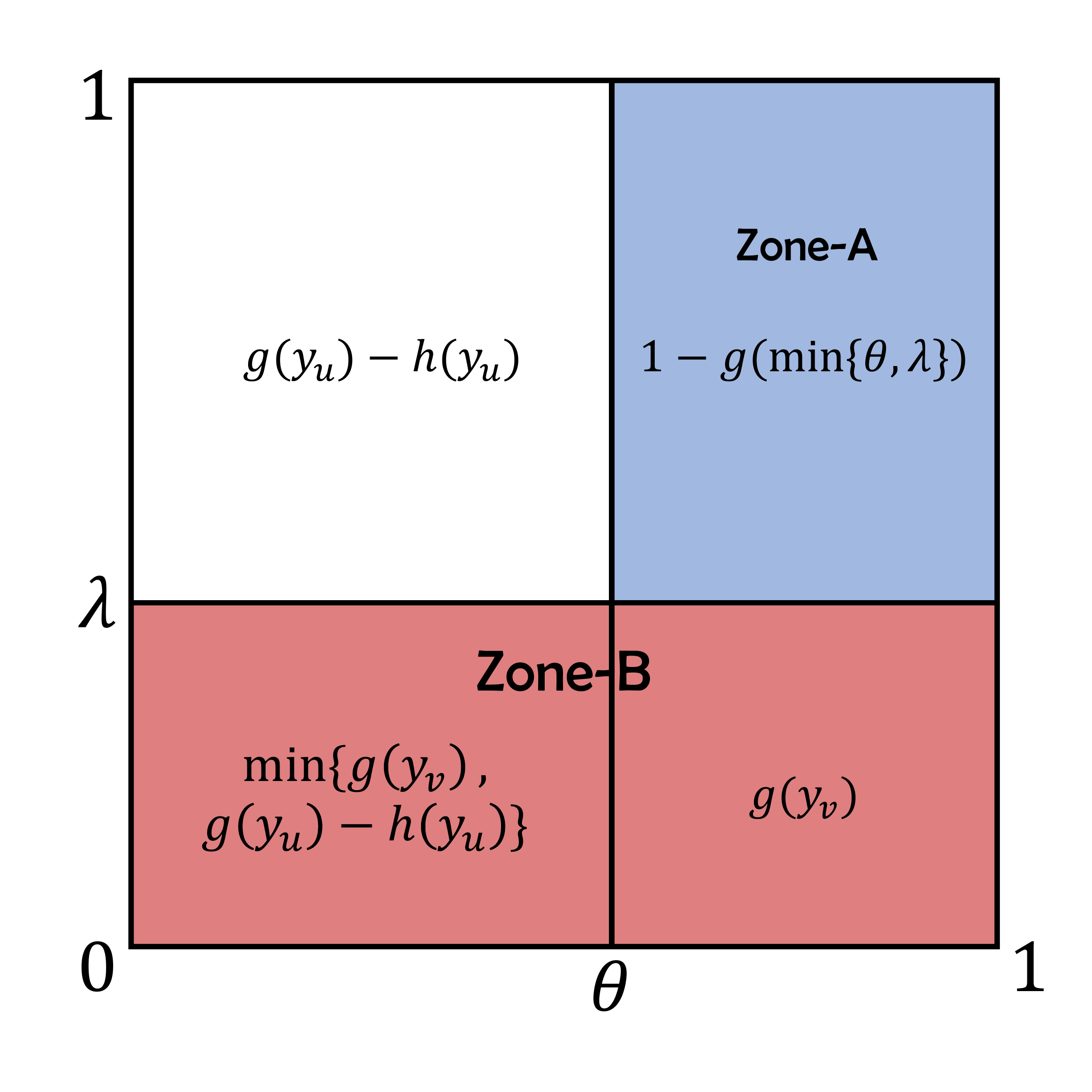}
		\vspace*{-10pt}
		\caption{Basic lower bounds}
		\label{fig:weighted_basic}
	\end{subfigure}%
	\begin{subfigure}[H]{.3\textwidth}
		\centering
		\includegraphics[width=0.9\linewidth]{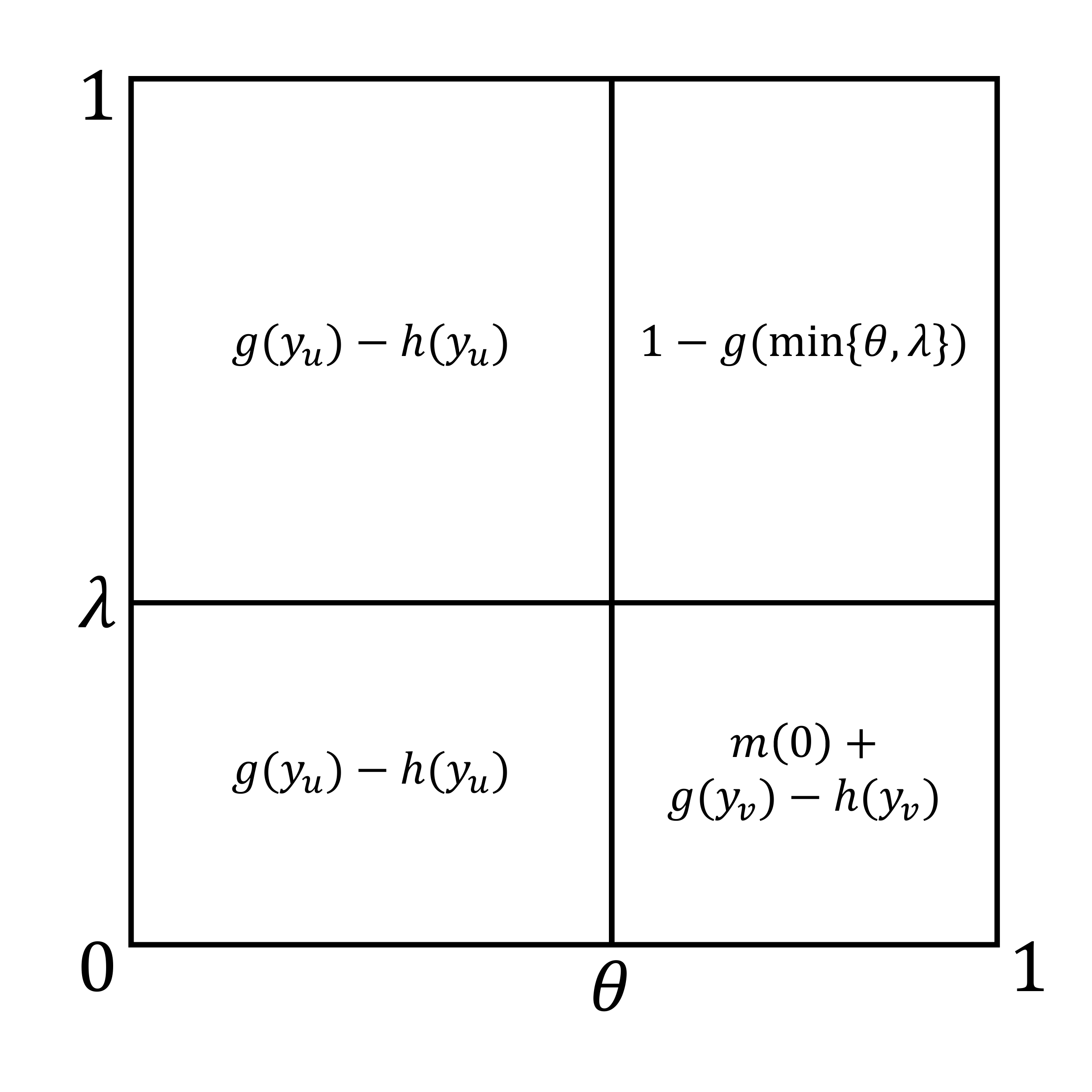}
		\vspace*{-10pt}
		\caption{$u$ strictly earlier in $M_{z}(1,1)$}
		\label{fig:we_ge_u_better}
	\end{subfigure}%
	\caption{Lower bounds on $\alpha_u + \alpha_v$.}
	\label{fig:we_ge_1}
\end{figure}

\subsection{Asymmetric Case: $(u,v)\notin M(1,1)$}

Without loss of generality, suppose $u$ is matched strictly earlier than $v$ in $M(1,1)$ by $z$.
Let $\theta$ be the transition rank such that $u$ is active in $M(y_u,1)$ when $y_u < \theta$ and passive when $y_u > \theta$.
Note that when $y_u > \theta$, $u$ is always passively matched by $z$ in $M(y_u,1)$.
Let $\lambda$ be the transition rank such that $v$ is matched earlier than $u$ in $M(1,y_v)$ when $y_v < \lambda$ and later than $u$ when $y_v > \lambda$.

Notably, $\lambda=\theta=y_z$ in the unweighted case, while all three parameters might differ for weighted graphs.
For example, suppose $u$ has another neighbor $x$ and $\theta$ is the critical rank such that the perturbed weight of edge $(u,x)$ beats the perturbed weight of the edge $x$ is matched with in $M(1,1)$.

\medskip

First of all, $u$ is matched by $z$ when $y_u > \theta$ and $y_v > \lambda$, which means that $\alpha_u = (1-g(y_z))w_{uz}$.
Since edge $(u,z)$ is probed earlier than $(u,v)$ when $y_u > \theta$, we have $(1-g(y_z))w_{uz} \ge (1-g(\theta))w_{uv}$. Similarly $(1-g(y_z))w_{uz} \ge (1-g(\lambda))w_{uv}$.
Thus $\alpha_u \geq 1-g(\min\{\theta,\lambda\})$ when $y_u > \theta$ and $y_v > \lambda$.

When $y_u < \theta$ and $y_v > \lambda$, we know that $u$ is matched earlier than $v$, as otherwise $v$ would be matched strictly earlier than $u$ in $M(1,y_v)$, violating the definition of $\lambda$. Moreover, $u$ must be active, as otherwise it is also passive in $M(y_u,1)$, violating the definition of $\theta$.
Then by Fact~\ref{fact:matched_ealier}, we have $\alpha_u \geq g(y_u)-h(y_u)$.
Similarly, when $y_u > \theta$ and $y_v < \lambda$, $v$ must be active and matched earlier than $u$.
Observe that $w_{uz} \ge \frac{1-g(1)}{1-g(y_z)} > \frac{1}{2}$. Applying Lemma~\ref{lemma:no-compen}, we have $\alpha_u+\alpha_v \ge g(y_v)$.

Finally, when $y_u <\theta$ and $y_v < \lambda$, the vertex matched earlier is active.
Moreover, since $u$ is matched strictly earlier than $v$ in $M(1,1)$, the edge $u$ matches in every $M_{v}(y_u,1)$ has weight at least $w_{uz}\geq \frac{1}{2}$.
Hence by Lemma~\ref{lemma:no-compen}, for all $y_u <\theta$ and $y_v < \lambda$, if $v$ is matched earlier then we have $\alpha_u + \alpha_v \geq g(y_v)$.
By Fact~\ref{fact:matched_ealier}, if $u$ is matched earlier we have $\alpha_u \ge g(y_u)-h(y_u)$. Thus, $\alpha_u + \alpha_v \ge \min \{g(y_u) -h(y_u), g(y_v)\}$.
Note that symmetrically, if we can show the edge $v$ matches in every $M_{u}(1,y_v)$, where $y_v < \lambda$, has weight at least $\frac{1}{2}$, then by applying Lemma~\ref{lemma:no-compen} we can improve the lower bound to $\alpha_u + \alpha_v \ge \min \{g(y_u), g(y_v)\}$.

In summary, we have the following lower bounds that serve as our basic gains (refer to Fig.~\ref{fig:weighted_basic}).
\begin{enumerate}
	\item[(L1)]  $\alpha_u \ge 1-g(\min\{\theta,\lambda\})$ when $y_u > \theta$ and $y_v > \lambda$;
	\item[(L2)] $\alpha_u  \geq g(y_v) - h(y_u)$ when $y_u < \theta$ and $y_v > \lambda$;
	\item[(L3)] $\alpha_u+\alpha_v \ge g(y_v)$ when $y_u > \theta$ and $y_v < \lambda$;
	\item[(L4)] $\alpha_u + \alpha_v \ge \min \{g(y_u) -h(y_u), g(y_v)\}$ when $y_u <\theta$ and $y_v < \lambda$.
\end{enumerate}

\paragraph{Extra Gains.}
Similar to our analysis for the unweighted case, we define \textsf{Zone-A} to be the matchings $M(y_u,y_v)$ when $y_u>\theta$ and $y_v>\lambda$ (where lower bound (L1) is applied), and \textsf{Zone-B} to be the matchings $M(y_u,y_v)$ when $y_v<\lambda$ (where lower bounds (L3) (L4) are applied).
In the following, we show that better lower bounds can be obtained for at least one of these bounds.
Roughly speaking, (like the unweighted case) if only one of $u,v$ is matched in \textsf{Zone-B}, then we should be able to recover some extra gain in \textsf{Zone-A}, e.g., the compensation received by $v$.
We continue our analysis according to the matching status of $u$ and $v$ in $M_{z}(1,1)$, i.e., when $z$ is removed from the graph.

\subsubsection{Case 1: $(u,v) \notin M_{z}(1,1)$}
In this case, at least one of $u,v$ is passively matched by other vertices in $M_{z}(1,1)$.
We first consider the case when $u$ is matched strictly earlier.
In such case, we show that $u$ has gain at least $m(y_u)$ in \textsf{Zone-B}, as it has a ``backup'' neighbor other than $z,v$.
The following lemma is a weighted version of Lemma~\ref{lemma:uw_ge_uearlier}, and the proof is similar. We formalize it in a general way so that we can also apply it in later analysis.

\begin{lemma} \label{lemma:weighted_backup}
	When $y_u > \theta$ and $y_v < \lambda$, if $u$ is matched strictly earlier than $v$ in $M_{z}(y_u, 1)$, we have $\alpha_u \ge m(y_u)$ in $M(y_u,y_v)$.
\end{lemma}
\begin{proof}
	First, by definition $u$ is passively matched by $z$ in $M_{v}(y_u,y_v)$. Now consider $M(y_u,y_v)$.
	
	If the insertion of $v$ does not change the matching status of $u$, i.e., $u$ remains passively matched with $z$, then we have $\alpha_u \geq 1-g(\min\{\theta,\lambda\})\geq m(y_u)$.
	
	If $u$ is matched even earlier after the insertion, then the edge $(u,x)$ vertex $u$ matches has perturbed weight larger than $1-g(\min\{\theta,\lambda\})$.
	Hence if $u$ is passive, $\alpha_u \geq 1-g(\min\{\theta,\lambda\}) \geq m(y_u)$; if $u$ is active, $\alpha_u \geq (g(y_u)-h(y_u))\frac{1-g(\min\{\theta,\lambda\})}{1-g(y_u)} \geq g(y_u)-h(y_u)\ge m(y_u)$.
	
	Otherwise $u$ appears in the alternating path triggered by the insertion of $v$ as a vertex $v_{2i}$, which is matched later after the insertion (refer to Lemma~\ref{lemma:weighted_alternating}).
	Note that in this case, the vertex $v_{2i-1}$ right before $v_{2i} = u$ must be $z$.
	Hence the matching status of $u$ is not changed if we remove both $v$ and $z$ simultaneously in $M(y_u,y_v)$ (following the same argument as in the proof of Lemma~\ref{lemma:uw_ge_uearlier}).
	On the other hand, by Fact~\ref{fact:matched_ealier}, we have $\alpha_u \geq m(y_u)$ in $M_{z}(y_u,1)$.
	Moreover, the matching status of $u$ is not changed if we further remove $v$ in $M_{z}(y_u,1)$.
	Thus we have $\alpha_u \geq m(y_u)$ in $M(y_u,y_v)$.
\end{proof}

Note that if $u$ is matched strictly earlier than $v$ in $M_{z}(1,1)$, then by above lemma, for all $y_v < \lambda$, the gain of $u$ in $M(1,y_v)$ is at least $m(1)$.
Now consider decreasing the rank of $u$.
By Lemma~\ref{lemma:monotonicity}, when $u$ is passive, the gain of $u$ does not change until it becomes active, which gives $\alpha_u \ge 1-g(y_u)$; when $u$ is active, decreasing $u$'s rank would not decrease the edge weight it matches, and hence $\alpha_u \ge g(y_u)-h(y_u)$. 
For ease of analysis, we choose function $g$ such that $m$ achieves its minimum at $m(0)$.
To sum up, we have $\alpha_u \geq m(0)$ in $M(y_u,y_v)$ for all $y_v < \lambda$.

Therefore, we improve (L3) to $\alpha_u+\alpha_v \geq m(0)+g(y_v)-h(y_v)$ and (L4) to $\alpha_u + \alpha_v \ge \min \{g(y_u) -h(y_u), m(0)+g(y_v)-h(y_v)\}$.
Note that by restricting $g(y)\in [0.4,0.6]$, $m(0)+g(y_v)-h(y_v)$ is larger than $g(y_u) - h(y_u)$ for all $y_u,y_v$.
Consequently, we have the following (refer to Figure~\ref{fig:we_ge_u_better}). 
\begin{equation} \label{eq:we_ge_ubetter}
	\begin{split}
		\expect{y_u,y_v}{\alpha_u+\alpha_v}\ge & \int_0^\theta \Big( g(y_u)-h(y_u) \Big) dy_u + (1-\theta) \int_0^\lambda \Big(g(y_v) - h(y_v)\Big) dy_v \\
		& + (1-\theta)(1-\lambda)(1-g(\min\{\lambda,\theta\})) + (1-\theta)\lambda\cdot m(0). 
	\end{split}
\end{equation}

It remains to consider the case when $v$ is matched strictly earlier than $u$ in $M_{z}(1,1)$.
As we will show later, this case can actually be regarded as a ``better'' case of $(u,v)\in M_z(1,1)$, and thus we defer its analysis to the next subsection (see Remark~\ref{remar:vbetter}).

\subsubsection{Case 2: $(u,v) \in M_{z}(1,1)$} \label{subsec:weighted_case_2_2}

Finally, we consider the case when $u,v$ match each other in $M_{z}(1,1)$.
By definition, $v$ is the victim of $z$ in $M(1,1)$, and hence
\begin{equation*}
	\textstyle \alpha_v \geq h(y_z)w_{uz} = \frac{1}{10}(1-g(y_z))w_{uz} \geq \frac{1}{10}(1-g(\min\{\theta,\lambda\}))=h(\min\{\theta,\lambda\}).
\end{equation*}

We remark that this is the major reason for restricting $h=\frac{1}{10}(1-g)$, as otherwise we do not have an immediate connection between $h(y_z) w_{uz}$ and $h(\min\{\theta,\lambda\})$.
Next, we show that $v$ receives this compensation in (part of) \textsf{Zone-A}.

\begin{lemma}[Compensation] \label{lemma:compensation}
	For all $y_u>\theta$ and $y_v >\lambda$, if $u$ matches $v$ in $M_{z}(y_u,1)$, then we have $\alpha_v \geq h(\min\{\theta,\lambda\})$ in $M(y_u,y_v)$.
\end{lemma}
\begin{proof}
	Consider $M_z(y_u,y_v)$. If $u,v$ match each other, then $v$ is the victim of $z$ in $M(y_u,y_v)$ and the lemma holds.
	Otherwise we know that $v$ must be matched strictly earlier than $u$: if $u$ is matched strictly earlier, than it will also be matched strictly earlier than $v$ in $M_z(y_u,1)$, which contradicts the assumption that $u$ matches $v$ in $M_z(y_u,1)$.
	
	Thus we have $\alpha_v \geq m(y_v)$ in $M_z(y_u,1)$.
	Moreover, if we further remove $u$ in $M_z(y_u,1)$, the matching status of $v$ is not changed.
	Since $z$ matches $u$ in $M(y_u,y_v)$, removing both $z$ and $u$ does not change the matching status of $v$, which means that $\alpha_v \geq m(y_v) > h(\min\{\theta,\lambda\})$ in $M(y_u,y_v)$.
\end{proof}

The above lemma indicates that it is crucial to determine whether $u$ matches $v$ in $M_z(y_u,1)$ when $y_u >\theta$.
Recall that $(u,v)\in M_z(1,1)$.
We define $\gamma \in [\theta,1]$ to be the transition rank such that $u$ matches $v$ in $M_{z}(y_u,1)$ when $y_u \in (\gamma,1)$ and matches other vertex when $y_u\in (\theta, \gamma)$.

\medskip

We first consider the case that $u$ matches $v$ in $M_z(y_u,1)$ for all $y_u >\theta$.
In this case $\gamma = \theta$, and thus by Lemma~\ref{lemma:compensation}, we have $\alpha_u + \alpha_v \geq 1-g(\min\{\theta,\lambda\})+h(\min\{\theta,\lambda\})$ in \textsf{Zone-A}.
The gain in \text{Zone-B} is more complicated. In particular, we need to consider whether $u,v$ match each other in $M(1,y_v)$ when $y_v < \lambda$.
Let $\tau\in [0,\lambda]$ be the transition rank such that $v$ matches $u$ in $M(1,y_v)$ when $y_v \in (\tau,\lambda)$ and matches strictly earlier than $u$ when $y_v\in (0,\tau)$. 

\vspace*{-10pt}
\begin{figure}[H]
	\centering
	\begin{subfigure}{.32\textwidth}
		\centering
		\includegraphics[width=0.9\linewidth]{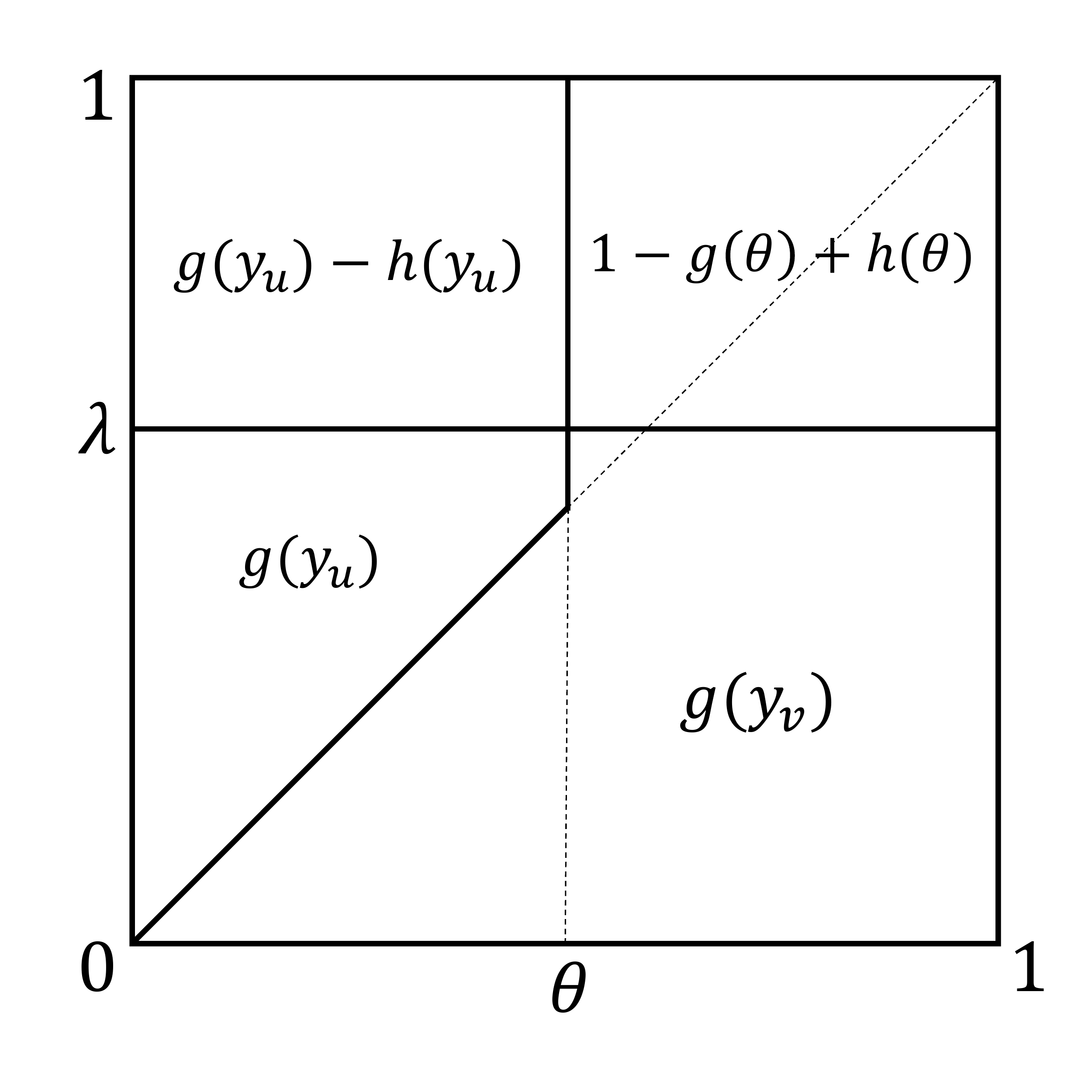}
		\vspace*{-10pt}
		\caption{$\gamma=\theta$ and $\lambda > \theta$}
		\label{fig:we_ge_gamma1}
	\end{subfigure}%
	\begin{subfigure}{.32\textwidth}
		\centering
		\includegraphics[width=0.9\linewidth]{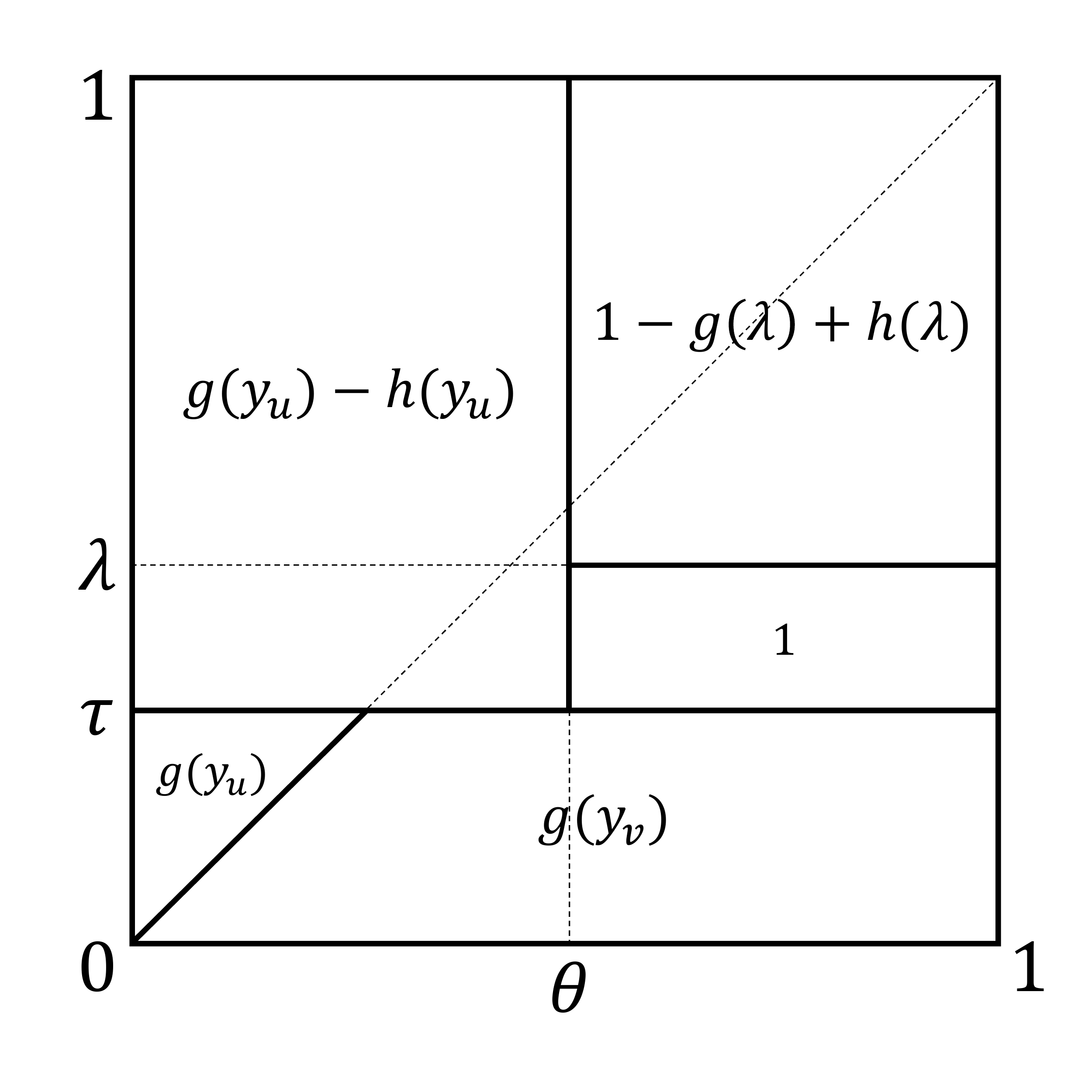}
		\vspace*{-10pt}
		\caption{$\gamma=\theta$ and $\lambda \le \theta$}
		\label{fig:we_ge_gamma2}
	\end{subfigure}
	\begin{subfigure}{.32\textwidth}
		\centering
		\includegraphics[width=0.9\linewidth]{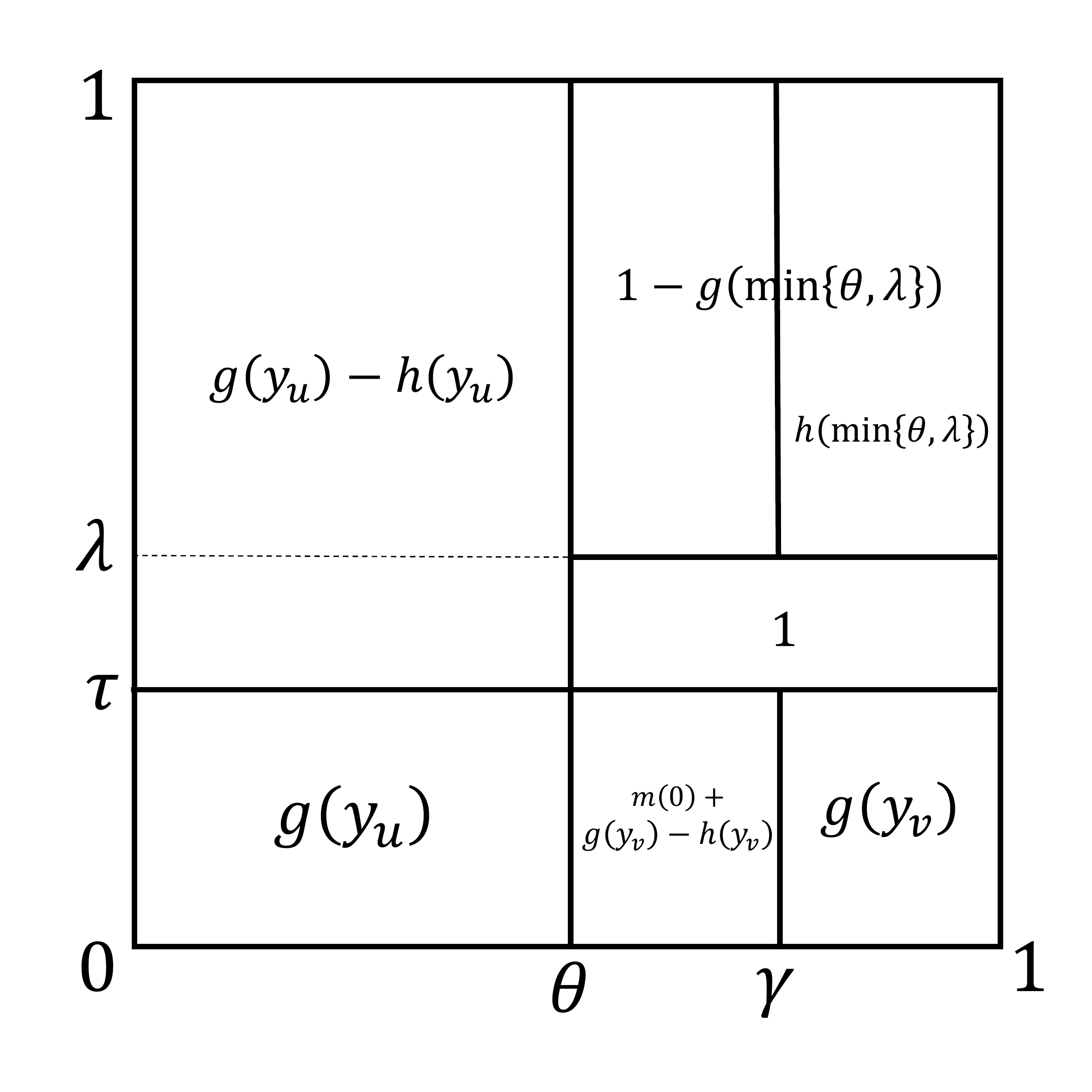}
		\vspace*{-10pt}
		\caption{$\gamma \in (\theta,1)$}
		\label{fig:we_ge_gamma3}
	\end{subfigure}%
	\caption{$(u,v)\in M_z(1,1)$.}
\end{figure}
\vspace*{-10pt}

To proceed, we compare $\theta$ and $\lambda$ and consider the following two cases.

\paragraph{When $\lambda > \theta$.}
First, observe that when $\lambda > \theta$, we must have $\tau = \lambda$.
In other words, $v$ never matches $u$ in $M(1,y_v)$ when $y_v < \lambda$.
The reason is, when $y_v = \lambda^-$, the perturbed weight of edge $(u,v)$ is $1-g(y_v) < 1-g(\theta) < (1-g(y_z))w_{zu}$. However, we know that $v$ is matched earlier than $z$. Thus, $v$ actively matches an edge with weight larger than $1$. Consequently by Lemma~\ref{lemma:monotonicity}, $v$ does not match $u$ in $M(1,y_v)$ for all $y_v < \lambda$. 

In other words, $v$ is matched in $M_u(1,y_v)$ with an edge of weight at least $1$ for all $y_v < \lambda$.
Thus we apply Lemma~\ref{lemma:no-compen} and improve (L4) to $\alpha_u + \alpha_v \geq \min\{g(y_u),g(y_v)\}$. To sum up, we have the following lower bound (refer to Figure~\ref{fig:we_ge_gamma1}).
\begin{equation}
	\begin{split}
		\expect{y_u,y_v}{\alpha_u+\alpha_v} \ge & (1-\theta)(1-\lambda)(1-g(\theta)+h(\theta)) + (1-\lambda) \int_0^\theta \Big( g(y_u)-h(y_u) \Big) dy_u \\
		& + (1-\theta) \int_0^\lambda g(y_v)dy_v + \int_0^\theta (\theta+\lambda-2y)g(y)dy.
	\end{split} 
\end{equation}

\paragraph{When $\lambda \le \theta$.}
In this case $v$ matches $u$ in $M(1,y_v)$ when $y_v \in (\tau,\lambda)$.
Moreover, $u$ remains matched by $v$ in $M(y_u,y_v)$ when $y_u>\theta$ and $y_v\in (\tau,\lambda)$.
When $y_v < \tau$, we have $\alpha_u + \alpha_v \geq \min\{g(y_u),g(y_v)\}$, by a similar argument as before.
In summary, we have (refer to Figure~\ref{fig:we_ge_gamma2})
\begin{equation}
	\begin{split}
		\expect{y_u,y_v}{\alpha_u+\alpha_v} \ge & (1-\theta)(1-\lambda)(1-g(\lambda)+h(\lambda)) + (1-\tau) \int_0^\theta \Big( g(y_u)-h(y_u) \Big) dy_u \\
		& + (1-\theta)(\lambda-\tau) +(1-\theta) \int_0^\tau g(y_v)dy_v + \int_0^\tau (\theta+\tau-2y)g(y)dy.
	\end{split} 
\end{equation}

The derivative over $\lambda$ of the RHS of above is $(1-\theta)\big(g(\lambda)-h(\lambda)-(1-\lambda)(g'(\lambda)-h'(\lambda))\big)$.
Our choice of $g,h$ guarantees that $g(y)-h(y)$ is always larger than $g'(y)-h'(y)$, and thus this derivative is positive.
Therefore, the minimum is achieved when $\lambda=\tau$.

To sum up, the two lower bounds can be unified as follows.
\begin{equation} \label{eq:we_ge_comp}
	\begin{split}
		\expect{y_u,y_v}{\alpha_u+\alpha_v} \ge & (1-\theta)(1-\lambda)(1-g(\min\{\theta,\lambda\})+h(\min\{\theta, \lambda\}))  + (1-\theta) \int_0^\lambda g(y_v)dy_v\\
		& +(1-\lambda)\int_0^\theta \Big( g(y_u)-h(y_u)\Big) dy_u
		+ \int_0^{\min\{\theta,\lambda\}} (\lambda+\theta-2y) g(y)dy.
	\end{split}
\end{equation}

\begin{remark}\label{remar:vbetter}
	Now if we turn our attention back to the case when $v$ is matched strictly earlier than $u$ in $M_z(1,1)$, we can see that~\eqref{eq:we_ge_comp} also serves as a lower bound.
	For \textsf{Zone-A}, the lower bound $1-g(\min\{\theta,\lambda\})+h(\min\{\theta,\lambda\})$ holds since $v$ is matched strictly earlier than $u$ in $M_z(1,y_v)$ when $y_v > \lambda$, which implies $\alpha_v \geq m(y_v)$.
	Exactly the same analysis on lower bounds for \textsf{Zone-B} can also be applied to this case.
\end{remark}

Finally, we consider the case when $\gamma > \theta$.
By the definition of $\gamma$ and Lemma~\ref{lemma:compensation}, $\alpha_v \ge h(\min\{\theta,\lambda\})$ in the part of \textsf{Zone-A} when $y_u > \gamma$.
When $y_u>\theta$ and $y_v\in (\tau,\lambda)$, $u,v$ match each other and $\alpha_u + \alpha_v = 1$.
It remains to give lower bounds when $y_v < \tau$.

Since $v$ is matched strictly earlier than $u$ in $M(1,y_v)$, Lemma~\ref{lemma:no-compen} applies and we have $\alpha_u + \alpha_v \geq \min\{ g(y_u),g(y_v) \}$ for all $y_u\in (0,1)$ and $y_v < \tau$.
Furthermore, since $\gamma > \theta$, we can actually improve this bound further.
First, by Lemma~\ref{lemma:weighted_backup}, we have $\alpha_u \ge m(y_u) \geq m(0)$ when $y_u \in (\theta, \gamma)$ and $y_v < \tau$.
Second, when $y_u < \theta$ and $y_v < \tau$, we can improve (L3) to $\alpha_u + \alpha_v \geq \min\{ g(y_u),g(y_v)-h(y_v)+m(y_u) \} = g(y_u)$.
In summary, we have the lower bounds as shown in Figure~\ref{fig:we_ge_gamma3}.
\begin{equation} 
	\label{eq:we_gamma}
	\begin{split}
		\expect{}{\alpha_u+\alpha_v} \ge & (1-\theta)(1-\lambda)(1-g(\min\{\theta,\lambda\})) + (1-\gamma)(1-\lambda)h(\min\{\theta,\lambda\}) \\
		& +  \int_0^\theta \Big( g(y_u)-(1-\tau) h(y_u) \Big) dy_u
		+ \int_0^\tau \Big( g(y_v) - (\gamma-\theta) h(y_v) \Big) dy_v \\
		& + (1-\theta)(\lambda-\tau) + (\gamma-\theta)\tau \cdot m(0).
	\end{split}
\end{equation}

It is straightforward to see that for every fixed $\tau$, the minimum of the lower bound is achieved when $\lambda=\tau$ (given that $1-g(y)+h(y)$ is always smaller than $1$).
Moreover, for every fixed $\lambda = \tau$, the derivative over $\gamma$ is a constant.
Hence the minimum must be achieved when $\gamma\in\{\theta,1\}$.
It is easy to check that when $\gamma=1$, Equation~\eqref{eq:we_ge_ubetter} serves as an lower bound for Equation~\eqref{eq:we_gamma}; when $\gamma=\theta$, Equation~\eqref{eq:we_ge_comp} serves as an lower bound.

\subsection{Lower Bounding the Approximation Ratio}

Unlike the unweighted case, the performance of \textsf{Perturbed Greedy} on edge-weighted graphs depends on the design of the function $g$. To this end, we explicitly construct one and analytically prove the approximation ratio, rather than running a factor revealing LP.
In particular, we construct a function $g$ that satisfies all pre-specified constraints such that the lower bounds~\eqref{eq:weighted_general_uvmatch}~\eqref{eq:we_ge_ubetter}~\eqref{eq:we_ge_comp} are at least $0.5014$, which completes the proof of Theorem~\ref{th:weighted_general}.
We remark that the piece-wise linear function is just an artifact of the proof, and is not optimal for maximizing the approximation ratio.

In the following, we fix function $g(y) = \begin{cases}
0.365y + 0.48926, & y\leq 0.13 \\
0.067y + 0.528, & y\in (0.13,0.4) \\
0.5548 & y\geq 0.4.
\end{cases}$

First, it is easy to see that the constraints we put on $g$ are satisfied: for every $y\in (0,1)$, we have $g(y)\in (0.4,0.6)$ and $g(y)-h(y) \geq g'(y)-h'(y)$.
It is also easy to check that for the function $g$ we fix, $\min_y\{m(y)\} = m(0) = g(0)-h(0) = 0.438186$.

\subsubsection{Equation~\eqref{eq:weighted_general_uvmatch}}\label{sssec:analysis_1}

We prove that the RHS of~\eqref{eq:weighted_general_uvmatch} (shown as follows) is at least $0.5014$ (over $\theta\geq \lambda$).
\begin{align*}
	& (1-\theta)(1-\lambda) + (1-\lambda) \int_0^\theta \Big( g(y_u)-h(y_u) \Big)dy_u + \int_0^\lambda (\lambda-y_u)g(y_u) dy_u \\
	& + (1-\theta) \int_0^\lambda \Big( g(y_v) - h(y_v) \Big) dy_v + \int_0^\lambda (\theta-y_v)g(y_v) dy_v.
\end{align*}

First, if we take derivative over $\theta$, we have
\begin{equation*}
	(1-\lambda)(g(\theta)-h(\theta)-1)+\int_0^\lambda h(y)dy,
\end{equation*}
which is negative (for any $\theta$) when $\lambda \leq 0.9$.

Thus for $\lambda\leq 0.9$, the minimum is achieved when $\theta = 1$:
\begin{equation*}
	(1-\lambda)\int_0^1 \Big(g(y)-h(y)\Big) dy +\int_0^\lambda (\lambda-y)g(y) dy
	+ \int_0^\lambda (1-y)g(y) dy.
\end{equation*}

By taking derivative over $\lambda$, we have
\begin{equation*}
	\int_0^\lambda g(y)dy + (1-\lambda)g(\lambda) - \int_0^1 \Big( g(y)-h(y) \Big) dy.
\end{equation*}

Since this derivative is non-decreasing, the minimum is achieved when $\lambda = \lambda^* = 0.0344402$ (solution for $\int_0^\lambda g(y)dy + (1-\lambda)g(\lambda) = \int_0^1 g(y)-h(y) dy$), and the value is at least $0.5014$.

For $\lambda > 0.9$, we relax $(1-\theta)(1-\lambda)$ to be $(1-\lambda) \int_\theta^1 \Big( g(y_u)-h(y_u) \Big)dy_u$, then we have
\begin{equation*}
	(1-\lambda) \int_0^1 \Big( g(y_u)-h(y_u) \Big)dy_u + \int_0^\lambda (\lambda-y_u)g(y_u) dy_u
	+ (1-\theta) \int_0^\lambda \Big( g(y_v) - h(y_v) \Big) dy_v + \int_0^\lambda (\theta-y_v)g(y_v) dy_v,
\end{equation*}
which attains its minimum when $\theta = \lambda$: (the inequality holds since $h$ is non-increasing)
\begin{align*}
	& (1-\lambda) \int_0^1 \Big( g(y)-h(y) \Big)dy + (1-\lambda) \int_0^\lambda \Big( g(y) - h(y) \Big) dy + 2 \int_0^\lambda (\lambda-y)g(y) dy \\
	\geq  & (1-\lambda) \int_0^1 \Big( g(y)-h(y) \Big)dy + (1 + \lambda) \int_0^\lambda g(y) dy - 2 \int_0^\lambda y\cdot g(y) dy - (1-\lambda)\lambda \cdot h(0).
\end{align*}

Observe that the derivative (for $\lambda \geq 0.9$)
\begin{equation*}
	\int_0^\lambda g(y) dy + (1-\lambda) g(\lambda) - \int_0^1 \Big( g(y)-h(y) \Big) dy - (1-2\lambda)h(0) > 0.
\end{equation*}

Thus the minimum is achieved when $\lambda = 0.9$, which is at least $0.53$.

\subsubsection{Equation~\eqref{eq:we_ge_ubetter}}

We prove the RHS of \eqref{eq:we_ge_ubetter} (shown as follows) is at least $0.5016$ (over all $\theta$ and $\lambda$):
\begin{align*}
	&  \int_0^\theta \Big( g(y_u)-h(y_u) \Big) dy_u + (1-\theta)\int_0^\lambda \Big(g(y_v) - h(y_v)\Big) dy_v \\
	& + (1-\theta)(1-\lambda)(1-g(\min\{\lambda,\theta\})) + (1-\theta) \lambda\cdot m(0).
\end{align*}

First observe that if $\lambda > \theta$, then the derivative over $\lambda$ is non-negative, which means that the minimum in this case is achieved when $\lambda = \theta$.
Thus it suffices to consider the case when $\lambda \leq \theta$. For this case, the derivative over $\theta$ is given by
\begin{equation*}
	g(\theta) - h(\theta) - (1-\lambda)(1-g(\lambda)) - \lambda\cdot m(0) - \int_0^\lambda \Big( g(y)-h(y) \Big) dy.
\end{equation*}

It can be verified (by taking another derivative over $\lambda$) that the maximum value of the above derivative is achieved when $\lambda = 0$:
\begin{equation*}
	g(\theta) - h(\theta) - (1-g(0)) \leq g(1)-h(1)-(1-g(0)) = 0.51028 - 0.51074 < 0.
\end{equation*}

Thus the minimum of \eqref{eq:we_ge_ubetter} is attained when $\theta = 1$:
\begin{equation*}
	\int_0^1 \Big( g(y)-h(y) \Big) dy \geq 0.5016.
\end{equation*}

\subsubsection{Equation~\eqref{eq:we_ge_comp}}

We prove the RHS of \eqref{eq:we_ge_comp} (shown as follows) is at least $0.5014$ (over all $\theta$ and $\lambda$):
\begin{align*}
	& (1-\theta)(1-\lambda)(1-g(\min\{\theta,\lambda\})+h(\min\{\theta, \lambda\})) + (1-\theta) \int_0^\lambda g(y_v)dy_v \\
	& + (1-\lambda)\int_0^\theta \Big( g(y_u)-h(y_u) \Big) dy_u
	+ \int_0^{\min\{\theta,\lambda\}} (\theta + \lambda - 2y) g(y)dy.
\end{align*}

First, observe that except for a $(1-\lambda)\int_0^\theta h(y) dy$ term, the lower bound is symmetric for $\theta$ and $\lambda$.
Moreover, if $\theta < \lambda$, then $(1-\theta)\int_0^\lambda h(y) dy > (1-\lambda)\int_0^\theta h(y) dy$, which means that by swapping the values of $\theta$ and $\lambda$, the lower bound decreases.
Hence it suffices to consider the case when $\theta\geq \lambda$.

When $\theta\geq \lambda$, the derivative over $\theta$ is given by
\begin{equation*}
	(1-\lambda)\Big( (g(\lambda)-h(\lambda)) + (g(\theta)-h(\theta)) - 1 \Big)
	= \frac{11}{10}(1-\lambda)\Big( (g(\lambda) + g(\theta)) - \frac{12}{11} \Big).
\end{equation*}

Let $\lambda_0\approx 0.12835$ be the solution for $g(\lambda)+ g(1) = \frac{12}{11}$.
Since $g(y)$ is non-decreasing, for $\lambda < \lambda_0$, we have $g(\lambda)+ g(1) < \frac{12}{11}$, which implies that the above derivative is negative. Hence the minimum is achieved when $\theta = 1$:
\begin{equation*}
	(1-\lambda)\int_0^1 \Big(g(y)-h(y)\Big) dy +\int_0^\lambda (1+\lambda-2y)g(y) dy.
\end{equation*}

Then following the same argument as in Section~\ref{sssec:analysis_1}, the minimum value is at least $0.5014$.

For $\lambda \geq \lambda_0$, the minimum is achieved when $\theta = \theta^*$, where $g(\lambda) + g(\theta^*) = \frac{12}{11}$.
Let $\lambda^* = 0.260516$ be the solution for $g(\lambda) = \frac{6}{11}$.
Note that we must have $\lambda \leq \lambda^* \leq \theta^*\leq 0.4$.

For $\lambda\in ( \lambda_0,0.13 ]$, by definition of function $g$, $g(\lambda) + g(\theta^*) = \frac{12}{11}$ is equivalent to
\begin{equation*}
	0.067\cdot \theta^* + 0.528 + 0.365\cdot \lambda + 0.48926 = \frac{12}{11}.
\end{equation*}

Plugging in $\theta^*$ and using $g(\lambda) = 0.365\cdot \lambda + 0.48926$, we can explicitly express the lower bound as a cubic function of $\lambda$, which achieves its minimum value $0.5026$ when $\lambda = 0.13$.

For $\lambda\in ( 0.13,\lambda^* ]$, we have $\theta^* = 2\lambda^* - \lambda$.
Plugging in $\theta^*$ and using $g(\lambda) = 0.067\cdot \lambda + 0.528$, we can explicitly express the lower bound as another cubic function of $\lambda$, which achieves its minimum value $0.50235$ when $\lambda \approx 0.204$.

Thus for all $\lambda,\theta\in [0,1]$, the lower bound is at least $0.5014$.

\section{Conclusion and Open Questions}

In this paper, we prove that \rdt breaks the $0.5$ barrier by using a random decision order and arbitrary preference orders. A natural question to ask is whether the random preferences help in $\mrg$, i.e., whether $\mrg$ has strictly larger approximation ratio than $\rdt$ in the worst case. We slightly believe so and would like to see techniques extending the current gain sharing framework to analyze random preference orders.

We also propose the first algorithm that achieves approximation ratio strictly greater than $0.5$ for the edge-weighted oblivious matching problem. 
Careful readers might wonder what is the approximation ratio of our algorithm when applied to edge-weighted bipartite graphs. 
Actually we are aware of a modified version of our algorithm that achieves $1-\frac{1}{e}$ approximation\footnote{We have a manuscript containing the proof. To avoid distraction, we decide not to include it in this paper.}, in which we only sample ranks on one side of the graph and then perturb the weight of each edge $(u,v)$ by a multiplicative factor $(1-g(y_u))$.
We conjecture that the \textsf{Perturbed Greedy} algorithm (with an appropriate choice of $g$) proposed in this paper has approximation ratio strictly greater than $1-\frac{1}{e}$ and we leave this as an open problem.

{
	\bibliography{matching}
	\bibliographystyle{plain}
}

\newpage

\appendix

\section{Missing Proofs from Section~\ref{sec:prelim}} \label{appendix:prelim}

\begin{proofof}{Lemma~\ref{lemma:monotonicity}}
	For every vertex $v$, let $E_v$ be the set of edges $(u,v)\in E$ such that $y_v < y_u$.
	Note that if we increase $y_v$, then the perturbed weights of edges in $E_v$ decrease.
	
	Suppose $v$ is passive.
	Then at the moment when $v$ is matched, if an edge in $E_v$ is probed, then the other endpoint must be matched already. 
	Hence when $y_v$ increases, all these probes remain unsuccessful, i.e., $v$ gets matched by the same vertex and nothing is changed to the matching.
	The case when $v$ is unmatched is similar.
	
	Consequently, imagine that we increase $y_v$ gradually from $0$ to $1$, then once $v$ becomes passive or unmatched, then the matching remains unchanged afterwards. Thus there exists threshold $y<y_v$ such that $v$ is active when $y_v\in(0,y)$; passive or unmatched otherwise.
	
	Finally, for the last argument, suppose $v$ actively matches $u$ when $y_v\in (0,y)$, then when $y_v$ decreases, perturbed weights of edges in $E_v$ increase.
	Thus when edge $(v,u)$ is probed, either $v$ is already matched (with some edge with a larger perturbed weight), or $v$ matches $u$.
	Since all edges in $E_v$ is perturbed by a factor of $1-g(y_v)$, for smaller value of $y_v\in (0,y)$, the weight of the edge $v$ matches is not smaller.
\end{proofof}

\begin{proofof}{Lemma~\ref{lemma:alternating}}
	Suppose $u$ is matched with $u_1$ in $M(\vecy)$, then after removing $u$, at decision time of edge $(u,u_1)$, $u_1$ is no longer matched.
	If $u_1$ is unmatched in $M_{u}(\vecy)$, then the symmetric difference is a single edge $(u,u_1)$ and the statements trivially hold.
	Otherwise let $u_2$ be matched with $u_1$ in $M_{u}(\vecy)$.
	Observe that the decision time of edge $(u_1,u_2)$ is no earlier than $(u,u_1)$, as otherwise $u_1$ will remain matched with $u_2$ in $M(\vecy)$.
	Then by induction on the number of remaining vertices, the symmetric difference between $M(\vecy)$ and $M_{u_2}(\vecy)$ is an alternating path $P$ starting from $u_2$ such that the decision times of edges are non-decreasing.
	Thus the symmetric difference between $M(\vecy)$ and $M_{u}(\vecy)$ is an alternating path starting from $u,u_1,u_2$, followed by path $P$.
	Moreover, the decision time of $(u,u_1)$ is no later than $(u_1,u_2)$, and the first edge of $P$ has decision time no earlier than $(u_1,u_2)$, which implies the first statement.
	For every vertex $u_i$ with odd index in the alternating path, since the decision time of $(u_{i-1},u_{i})\in M(\vecy)$ is no later than $(u_{i},u_{i+1})\in M_{u}(\vecy)$, compared with $M_{u}(\vecy)$, $u_{i}$ is matched no later in $M(\vecy)$.
\end{proofof}

\section{Factor Revealing LP} \label{sec:factor_revealing_lp}

Recall that to lower bound the approximation ratio, we need to define appropriate functions $g$ and $h$ such that the lower bounds we formulate on $\expect{}{\alpha_u + \alpha_v}$ are at least some ratio $r > 0.5$ for all parameters $\theta,\lambda,\tau$ and $\gamma$.

For instance, for the unweighted bipartite case\footnote{The same approach can be applied to formulate a factor revealing LP for the unweighted general case.}, we formulated the following two lower bounds:
\begin{align*}
	L^1(\theta,\tau) = & \int_0^\theta (1-y)g(y)dy + \int_0^\theta \min\{y,\tau\} m(y) dy + (1-\theta)(1-\theta+\tau)(1-g(\theta)) \\
	& + \int_0^\tau (1-y)g(y) dy + \int_0^\tau y m(y) dy + \frac{1}{2}(2-\tau-\theta)(\theta-\tau),\\
	L^2(\theta,\tau) = & \int_0^\theta (1-y)g(y)dy+ \int_0^\theta \min\{y,\tau\} m(y) dy + \frac{1}{2}(1-\theta)^2 \\
	& + \int_0^\tau (1-y)g(y) dy + \int_0^\tau y m(y) dy + \frac{1}{2}(1-\tau)^2.
\end{align*}

To prove Theorem~\ref{th:unweighted_bipartite}, it remains to find function $g$ such that
\begin{equation*}
	\min_{\theta,\tau}\{ L^1(\theta,\tau), L^2(\theta,\tau) \} \geq 0.639.
\end{equation*}

This can be proved by solving the following continuous optimization problem and showing that the optimal solution is at least $0.639$.
\begin{align*}
	\max. \quad & r \\
	\text{s.t.} \quad & r \leq L^1(\theta,\tau), \quad \forall \theta,\tau \in [0,1] \\
	& r \leq L^2(\theta,\tau), \quad \forall \theta,\tau \in [0,1]\\
	& g(y) \in [0,1], \quad \forall y\in[0,1] \\
	& g(y) \text{ is non-decreasing.}
\end{align*}

Since $L^1(\theta,\tau), L^2(\theta,\tau)$ are both linear in $g(y)$, we can discretize them as follows.
Let $n$ be the size of discretization.
The larger $n$ is, the more accurate we can solve for the above continuous optimization problem. 
Let $g$ be a step function such that $g(y) = g_i$ for all $y\in [ \frac{i}{n},\frac{i+1}{n})$.
By doing so, the integrations in the equations $L^1,L^2$ can be represented by linear summations over $g_i$'s.

Moreover, if we can formulate linear functions $D^1(i,j),D^2(i,j)$ in a way that $D^1(i,j) \leq L^1(\theta,\tau)$ and $D^2(i,j) \leq L^2(\theta,\tau)$ for all $\theta \in [ \frac{i}{n},\frac{i+1}{n} )$ and $\tau \in [ \frac{j}{n},\frac{j+1}{n} )$, then the optimal solution of the following finite LP provides a lower bound on the approximation ratio.
\begin{align*}
	\max. \quad & r \\
	\text{s.t.} \quad & r \leq D^1(i,j), \quad \forall i,j\in [n] \\
	& r \leq D^2(i,j), \quad \forall i,j\in [n] \\
	& g_{k-1} \leq g_{k}, \quad \forall k\in [n] \\
	& g_0 \geq 0,\quad g_n \leq 1. 
\end{align*}

We use $L^1(\theta,\tau)$ to illustrate how we derive the discretized relaxation $D^1(i,j)$.
$L^2(\theta, \tau)$ and other lower bounds derived in Section~\ref{sec:unweighted_gen} can be relaxed in a similar way. Thus, a similar factor revealing LP can be formulated for unweighted general graphs. We omit the tedious details. Our codes for solving the linear programmings are available upon request.

Suppose $\theta \in [ \frac{i}{n},\frac{i+1}{n} )$ and $\tau \in [ \frac{j}{n},\frac{j+1}{n} )$.
We consider the terms of $L^1(\theta,\tau)$ one by one and obtain the following lower bounds.
\begin{enumerate}
	\item $\int_0^\theta (1-y)g(y)dy \geq \frac{1}{n}\sum_{k=0}^{j-1} (1-\frac{k+1}{n})g_k$ and $\int_0^\tau (1-y)g(y) dy \geq \frac{1}{n}\sum_{k=0}^{i-1} (1-\frac{k+1}{n})g_k$;
	\item $\int_0^\theta \min\{y,\tau\} m(y) dy \geq \frac{1}{n}\sum_{k=0}^{i-1}\frac{\min\{k,j\}}{n}\cdot m_k$ and $\int_0^\tau y\cdot m(y) dy \geq \frac{1}{n}\sum_{k=0}^{j-1}\frac{k}{n}\cdot m_k$, where each $m_k$ is a new variable and we introduce two extra constraints $m_k \leq g_k$ and $m_k\leq 1-g_k$;
	\item $(1-\theta)(1-\theta+\tau)(1-g(\theta)) \geq (1-\frac{i+1}{n})(1-\frac{i+1}{n}+\frac{j}{n})(1-g_{i+1})$;
	\item $\frac{1}{2}(2-\tau-\theta)(\theta-\tau) \geq \frac{1}{2}(2-\frac{j+1}{n}-\frac{i+1}{n})(\frac{i}{n}-\frac{j+1}{n})$.
\end{enumerate}

Observe that $D^1(i,j)$ asymptotically approaches $L^1(\theta,\tau)$ when $n$ approaches infinity.
Hence the optimal value of the discretized program approaches that of the continuous program when $n\rightarrow \infty$.

\section{Hardness Results} \label{app:hardness}
\subsection{\rdt on Bipartite Graphs}
In this section, we construct a bipartite graph called \tbomb that is similar to the one in \cite{sicomp/ChanCWZ18}. We evaluate the average performance of \rdt on the given graph by experiments. By doing so, we suggest an experimental hardness result for \rdt. 

\begin{figure}[H]
	\centering
	\includegraphics[width=0.4\linewidth]{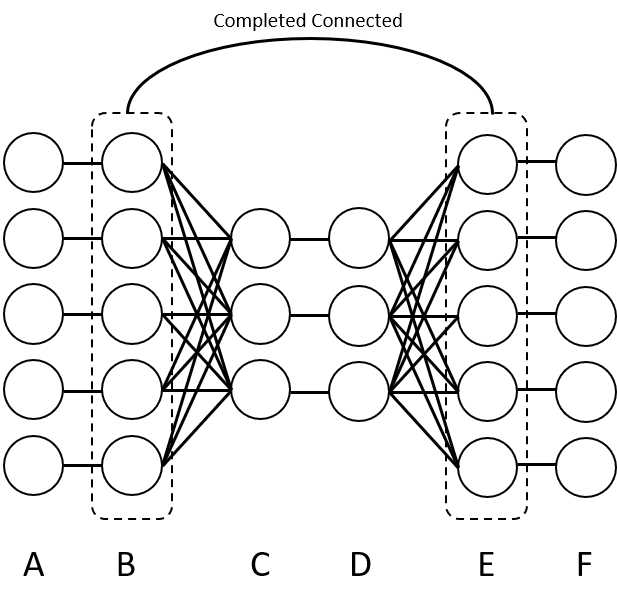}
	\caption{\tbomb}
	\label{fig:neg-bip}
\end{figure}

Refer to Figure~\ref{fig:neg-bip},  vertices of \tbomb consists of 6 parts $(A,B,C,D,E,F)$. $C,D$ contain $n_1$ vertices each, and $A,B,E,F$ contain $n_2$ vertices each. $n_1,n_2$ will be specified later in the experiment. 
The edges of the graph are defined as the following:
\begin{enumerate}
	\item $\forall i \in [n_1]$, let there be an edge $(C[i], D[i])$;
	\item $\forall i \in [n_2]$, let there be edges $(A[i],B[i])$ and $(E[i], F[i])$;
	\item $\forall i \in [n_1], j \in [n_2]$, let there be edges $(B[j],C[i])$ and $(D[i], E[j])$;
	\item $\forall i,j \in [n_1]$, let there be an edge $(B[i],E[j])$.
\end{enumerate}

Each group of vertices share the same preference order. Vertices in $A,F$ have only one neighbor, hence we don't need to specify the preferences. For vertices in $B$, they prefer vertices in $E$ to vertices in $C$ and finally to vertices in $A$. For vertices in $C$, they prefer vertices in $B$ to vertices in $D$. For vertices in $E$, they prefer vertices in $B$ to vertices in $D$ and finally to vertices in $F$. For vertices in $D$, they prefer vertices in $E$ to vertices in $C$. 
Here, within a group of vertices, the preference order is always from small index to large index. 

We run experiments for different $n_1,n_2$ (each for $10^5$ times), the average performance is shown in the following table. 
\begin{table}[H]
	\centering
	\begin{tabular}{|l|l|l|l|l|}
		\hline
		$n_1=$       & 100    & 200    & 500    & 1000   \\ \hline
		$n_2/n_1=1$   & 0.6514 & 0.6504 & 0.6499 & 0.6497 \\ \hline
		$n_2/n_1=1.3$ & 0.6479 & 0.6471 & 0.6465 & 0.6464 \\ \hline
		$n_2/n_1=1.5$ & 0.6474 & 0.6467 & 0.6461 & 0.646  \\ \hline
		$n_2/n_1=1.8$ & 0.6477 & 0.6471 & 0.6466 & 0.6465 \\ \hline
		$n_2/n_1=2$   & 0.6484 & 0.6478 & 0.6473 & 0.6471 \\ \hline
	\end{tabular}
\end{table}

We observe that the worst performance achieves when $n_2/n_1=1.5$, which is close to 0.646. We leave as future work to analyze it theoretically.

\subsection{\rdt  on General Graphs}
In this section, we construct a non-bipartite graph with $4$ vertices for which the maximum matching matches $4$ vertices while the \rdt algorithm matches $2.5$ vertices in expectation.
In other words, the approximation ratio of \rdt on general graphs is at most $0.625$.
Together with Theorem~\ref{th:unweighted_bipartite}, we show a separation on the approximation ratio of \rdt on bipartite and general graphs.

\begin{theorem}
	\rdt is at most $0.625$-approximate for general graphs.
\end{theorem}
\begin{proof}
	Let the four vertices be $\{a,b,c,d\}$, and let there be $4$ edges: $(a,b), (a,c), (b,c)$ and $(c,d)$.
	Obviously there exists a perfect matching.
	
	Let the preference of all vertices be $c > b > a > d$.
	It is easy to check that unless $d$ has the earliest decision time, \rdt matches only one edge.
	Thus the expected size of the matching produced by \rdt is $\frac{5}{4}$ while the maximum matching has size $2$.
\end{proof}

\subsection{\irp}

\begin{theorem}
	\irp is no better than $0.5$-approximate, even for bipartite graphs. 
\end{theorem}
\begin{proof}
	We state the instance from~\cite{rsa/DyerF91}. 
	Let the set of vertices be $\{u_i,v_i\}_{i\in[n]}$, where $u_i$ is connected to $v_i$, for all $i\in[n]$ and there is a complete bipartite graph between $\{u_1,\ldots,u_{\frac{n}{2}}\}$ and $\{u_{\frac{n}{2}+1},\ldots,u_{n}\}$.
	If in the decision order, all $u_i$'s appear before $v_i$'s, it is easy to check that the approximation ratio is $0.5+o(1)$. We omit the formal proof.
\end{proof}

\end{document}